\documentclass[10pt]{article}
\usepackage{graphicx,lmodern,booktabs}
\usepackage[dvipsnames]{xcolor}
\usepackage{amsmath,amssymb,mathtools,upgreek,amsthm,multirow,multicol}
\usepackage[colorlinks=true]{hyperref}
\usepackage[T1]{fontenc}
\usepackage[backend=biber,style = apa ]{biblatex}
\DeclareLanguageMapping{british}{british-apa}

\newcounter{subeq}

\hypersetup{
linkcolor=BrickRed
,citecolor=Green
,filecolor=Mulberry
,urlcolor=NavyBlue
,menucolor=BrickRed
,runcolor=Mulberry
,linkbordercolor=BrickRed
,citebordercolor=Green
,filebordercolor=Mulberry
,urlbordercolor=NavyBlue
,menubordercolor=BrickRed
,runbordercolor=Mulberry
}

\DeclareCiteCommand{\citeauthor}
  {\boolfalse{citetracker}%
   \boolfalse{pagetracker}%
   \usebibmacro{prenote}}
  {\ifciteindex
     {\indexnames{labelname}}
     {}%
   Printtext[bibhyperref]{Printnames{labelname}}}
  {\multicitedelim}
  {\usebibmacro{postnote}}

	\usepackage[mathlines]{lineno}
\newtheorem{theorem}{Theorem}[section] 
\newtheorem{corollary}{Corollary}[theorem]
\newtheorem{lemma}[theorem]{Lemma}
\theoremstyle{definition}
\newtheorem{definition}{Definition}[section]

\newtheorem{remark}{Remark}
\newcommand\diff{\mathop{}\!\mathrm{d}}
\begin{document}
\title{An unbiased minimum variance non-parametric analytic and likelihood estimator for discrete and continuous score spaces}
\author{\href{mailto:landon.hurley@yale.edu}{Landon Hurley, PhD}}
\maketitle

\tableofcontents
\begin{abstract}
This manuscript develops a general purpose inner-product norm for the Kendall \(\tau\) and Spearman's \(\rho\), which operates as an unbiased MLE even in the presence of ties. We derive and prove the strict sub-Gaussianity of the Kemeny norm-space, thereby disproving conclusions developed by both \textcite{kendall1948} and \textcite{diaconis1977} as to the nature of the appropriate, finite sample, probability distribution and test statistics. A non-parametric MLE framework for all bivariate pairs is developed, thereby resolving an hypothesis of \textcite{olkin1994} concerning an exponential multivariate distribution for order statistics, by showing that for finite samples, the distribution is non-exponential. Non-parametric linear estimators are also constructed for the polychoric correlations and by extension, a linearly decomposable non-parametric multidimensional linear system of equations for non-parametric Factor Analysis is shown.
\end{abstract}

The Wishart distributions and procedure is the canonical sampling distribution of the maximum likelihood estimator of the covariance matrix of a \(p\)-variate normal distribution. This parametric restriction allows for the construction of finite sample efficient estimators which satisfy the maximum likelihood properties of the Lehmann-Scheff\'{e} theorem. In particular, that of the unbiased minimum variance estimators. Under the generalised linear model, conditional upon certain parametric assertions, the work of \textcite{nelder1972} allows for asymptotic linear functions to hold for non-linear but monotonic scores functions distributed along assuredly true probabilistic characterisations. 

To examine multivariate distributions of non-Gaussian variables over linear functions, we introduce an alternative Frobenius norm-space distance, the \textcite{kemeny1959} metric. The Kemeny metric is linear over every extended real (\(\overline{\mathbb{R}}^{n}\)) valued vector space, whose expectation is shown to be a Riesz-Fr\'{e}chet representation of the conventional median, which is both square-summable and continuous. Further, the probability distribution of said function is, for every finite positive integer \(n\), shown to be strictly sub-Gaussian, monotonically convergent,  and to satisfy both the Gauss-Markov and Lehmann-Scheff\'{e} theorems under the additional assumption of a common finite population which is uniformly sampled. This is extended to show that each bivariate correlation coefficient has a Beta-Binomial distribution for all finite \(n\).

Using this framework, we develop the conventional bivariate non-parametric estimators: \textcite{spearman1904} \(\rho\) and \textcite{kendall1938} \(\tau\) correlations and correct for the non-stochastically measurable (systematic) errors in the presence of ties. This correction is achieved by ignoring the symmetric group of order \(n\), \(\mathrm{S}_{n}\), and instead focusing upon the non-constant permutation space with repetitions, or ties, \(\mathcal{M}\) of order \(n^{n}-n\). Despite the familiarity to almost every statistical student at or beyond the high school level, the exact relationship of these alternative correlations and the definition of a probability distributions for the rejection of the Neyman-Pearson Null Hypotheses, are less clearly understood, especially for finite samples.  Traditionally, these techniques have relied upon combinatorial approximations of the population distribution to avoid the necessity of a bivariate normal distribution, and instead focusing upon the relative ranking within a sample. However, the restriction to the finite permutation space upon \(\mathrm{S}_{n}\) space results in additional computational complexities, due to the restrictive removal of the possibility of ties. In both multivariate analyses and their surjective linear functional relationships upon a sample, the presence of ties result upon finite measures excludes the possibility of a unique probability distribution. 

We develop in place a statistical framework based upon the Kemeny distance, which can be easily shown to possess a number of highly desirable properties for both the Spearman's \(\rho\) and Kendall's \(\tau\) parameter estimates. These include belonging to a Frobenius norm-space with cross-product for the estimation of both the Euclidean and Kemeny metric spaces upon the extended reals, satisfaction of the necessary properties of a Bregman divergence, and producing an unbiased continuous linear estimator of the median, subject only to the assumption of a uniformly sampled common population of variate score values. A necessary consequence of the strict sub-Gaussianity is given in that asymptotic normality does not guarantee stability: certain test distributions used for rank distributions almost surely do not follow a \(t\) distribution nor a unit normal distribution for finite samples. In both scenarios, the inappropriate use of test statistics which follow these expectations result in both invalid statistical acceptances and rejections of the null hypotheses inconsistent with the expectations of the Cram\'{e}r-von Mises minimum variance estimators. We demonstrate how to correct these errors, and specify the explicit necessary distributions with the Beta-Binomial distribution. 

These techniques are further developed to address the almost surely positive definite structure and duality of the Kemeny and Euclidean metric spaces, and we conclude by showing that the positive definite structure upon the dual Kemeny space can be used to solve linear system of equations upon the Euclidean metric. Applications of this work have been tested to demonstrate a valid Expectation-Maximisation algorithm, clustering, and multiple linear regression, with strong guaranteed performance even for misspecified linear Euclidean models in multiple regression.  

\section{Defining the Kemeny distribution and estimation functions}
Let us assume the existence of a sample matrix \(X \in \overline{\mathbb{R}}^{n \times p}\), indexed for \(i = 1,\ldots,n\) and \(j = 1,\ldots,p\) with restriction that \(n \gg p\). For such a vector space, we propose the following distance function, constructed upon the permutation space basis \(\kappa: \overline{\mathbb{R}}^{n} \to n \times n\), which is indexed \(k,l = 1,\ldots,n\) for each data matrix column \(j\): 

\begin{subequations}
\begin{equation}
\label{eq:kem_dist}
\rho_{\kappa}(X,Y) = \frac{n^{2}-n}{2} + \sum_{k,l=1}^{n} \kappa(X)_{kl}\odot\kappa(Y)^{\intercal}_{kl},
\end{equation}
\begin{equation}
\label{eq:kem_score}
\kappa_{kl}(X) = { 
\begin{dcases}
\: \sqrt{.5} & \text{if } x_{k} > x_{l}\\
\: -\sqrt{.5} & \text{if } x_{k} < x_{l}\\
\: 0 & \text{if } x_{k} = x_{l},\\
\end{dcases}
}\, k,l = 1,\ldots,n.
\end{equation}
\end{subequations}
The \(\kappa\) function maps an extended real vector onto a skew-symmetric matrix of order \(n \times n\), \(\kappa: \overline{\mathbb{R}}^{n} \to (n \times n)\). Each entry in said \(\kappa\) matrix is denoted by an entry in the \(k^{th}\) row and \(l^{th}\) column by element \(\kappa_{kl}\), upon which is performed the Hadamard product (\(\odot\)) between two such matrices.  The distance calculated by the linear combination of the pairwise distance orderings, i.e., the permutation representation, of any two elements upon vector \(X_{i}\), allows for permutations with ties (a space of density \(n^{n}-n\) rather than \(n!\)). The matrix transpose is denoted by the superscript \((\cdot)^{\intercal}\). 

As constructed, the expectation of the distance function is therefore at fixed point \(\frac{n^{2}-n}{2}\), which results from the inner-product of 0 linearly combined with the population expected distance. Using these specific positive real values for the \(\kappa\) matrix scoring, we obtain a linear distance for which no rescaling proportional to \(n\) is necessary, thereby allowing us to define the Hilbert space construction for arbitrary point of origin, without loss of generality. This satisfies the isometric embedding characterisation of any Hilbert space. 

Also, we note the immediate connection that with a maximum distance of \(\pm \frac{(n^{2}-n)}{2}\) obtainable upon all \(n^{n}\) permutations with ties. Dividing the centred distance Kemeny distance by the maximum distance results in a standard correlation like measure, wherein a distance of 0 denotes a correlation of 1, and the maximum distance produces a measure of -1. This definition is explicitly conformant with the vector space structure of a skew-symmetric matrix, where the inner-product is defined as follows: \[\langle{Ax,y}\rangle = -\langle{\kappa(X),A\kappa^{\intercal}(Y)\rangle}, \forall X,Y \in \overline{\mathbb{R}}^{n},\] and from which directly follows the standard inner-product vector norm correlation for the Kemeny sample space for the population:
\begin{equation}
\label{eq:kem_cor}
\langle X,Y\rangle_{\kappa} = \tau_{\kappa}(X,Y) = -\frac{2}{n^{2}-n}\sum_{k=1}^{n}\sum_{l=1}^{n} \kappa(X)_{kl}\odot\kappa^{\intercal}(Y)_{kl}, \{X,Y\} \in \overline{\mathbb{R}}^{n \times 1}.
\end{equation}
Under this construction, there are \(\binom{n}{2}\) distinct choices for equation~\ref{eq:kem_score}, equivalent to the \(\frac{n^{2}-n}{2}\) degrees of freedom for any skew-symmetric matrix; this therefore allows us to trivially normalise the signed distance measure to that of a correlation coefficient, entirely consistent with the standard definition of an inner-product norm upon a skew-symmetric vector space (equation~\ref{eq:kem_cor}). Thus, it is shown that the bivariate extended real vector space of length \(n\) is computationally achievable, with ties, upon an inner-product operator (satisficing the definition of a Hilbert space).  

\begin{lemma}
\label{lem:unbiased}
The Kemeny correlation (equation~\ref{eq:kem_cor}) is an unbiased estimator.
\end{lemma}
\begin{proof}
For the space of the Kemeny correlation, observe that upon the population \(\mathcal{M}\) is countable for all finite \(n\), that the inner-product of the even function space is the distance 0, and equivalently the correlation is 0. For \(n=2\), the space of \(\mathcal{M} = 2^{2} - 2 = 2\) permutations possess symmetric distances of \(\{\pm{a}\}\), which sum to 0, thereby presenting an unbiased estimator. By induction, allow a finite telescoping sequence of \(n_{i} \in \mathcal{M}\), from which follows both the finite expectation of 0 (by the closure under addition for the skew-symmetric matrix of equation~\ref{eq:kem_score} in the Kemeny metric space) and also the symmetry of the telescoping positive and negative distances, which arise by the even function nature of the \(\kappa\) function. Then for any finite \(n\), the sum and inner product of two random variables \(X,Y\) of length \(n\) in the population \(\mathcal{M}\) indexed by \(m\) is \[\lim_{m\to \mathcal{M}}E_{m}(\tau_{X,Y}) = \int_{m=1}^{\mathcal{M}}\sum_{k,l = 1}^{n} \kappa(X_{m})\odot\kappa^{\intercal}(Y_{m}) \diff{m}= 0,\] and thereby completes the proof that the Kemeny correlation is an unbiased estimator.
\end{proof}

% \begin{lemma}
% \label{lem:kem_bounded}
% The spectrum of the Kemeny distance is totally bounded.
% \end{lemma}
% \begin{proof}
% For any finite \(n\), consider the spectrum of the Kemeny metric, which lies on the interval of \([-\frac{n^{2}-n}{2},\frac{n^{2}-n}{2}]\), and is therefore of length \(n^{2}-n\). By the mean value theorem, any continuous bounded space (continuity is assured for any Banach norm-space) possesses points on the interior of this interval as well. As this linear space is totally bounded for all finite \(n\) upon a vector in the extended reals, the Kemeny metric is totally bounded.
% \end{proof} 

We now show that the Kemeny correlation estimator also satisfies the Gauss-Markov theorem, thereby providing a best linear unbiased estimator, which satisfies the Lehmann-Scheff\'{e} theorem, defined as follows: 

\begin{definition}
\label{def:gauss_markov}
The Gauss Markov theorem defines a set of requirements which when satisfied guarantee that the ordinary least squares estimate for regression coefficients provide the best linear unbiased estimate (BLUE) possible point estimates upon a sample. The five Gauss Markov conditions are:
\begin{multicols}{2}
\begin{enumerate}
    \item{Linearity: the parameters we are estimating using the OLS method must be themselves linear.}
    \item{Random: our data must have been randomly sampled from the population.}
    \item{Non-Collinearity: the regressors being calculated are not perfectly correlated.}
    \item{Exogeneity: the regressors aren't correlated with the error term.}
    \item{Homoscedasticity: no matter what the values of our regressors might be, the error of the variance is constant.}
\end{enumerate}
\end{multicols}
\end{definition}

\begin{theorem}
\label{thm:gauss-markov}
The Kemeny correlation is a Gauss-Markov estimator for any bivariate vector pair of length \(n\) which are independently sampled from a common population.
\end{theorem}
\begin{proof}
The linear function space follows by definition from the existence of a Banach norm space, which subsumes the properties of the Hilbert space as shown in equation~\ref{eq:kem_cor} (see also Theorem~\ref{lem:hilbert}). Unbiasedness follows from Lemma~\ref{lem:unbiased}, and Gramian positive definiteness follows from either the satisfaction of the Mercer condition, or as the finite sum of the squared totally bounded positive variances of the Kemeny metric space (Lemma~\ref{lem:kem_bounded}; \cite{schoenberg1938}); both conditions are equivalent. Exogeneity follows from the Riesz representation theorem for any Hilbert space, and random sampling follows by axiomatic assumption of the theorem. The homoscedasticity assertion follows by definition of the Kemeny variance \(\sigma^{2}_{\kappa}\), as there may only exist one set of permutations upon a common population, for a common population function, which is again true by axiomatic assumption. This completes the proof that the Kemeny correlation satisfies all necessary requirements of a Gauss-Markov estimator.
\end{proof}

The non-negative definiteness of each variable is given by the square of a skew-symmetric matrix (whose transpose provides a negation which is then squared), with 0 obtained only upon constant valued vectors, which would correspond to a degenerate distribution, and may therefore be otherwise excluded. The sum of any such sequence of \(n^{2}-n\) such numbers, as the skew-symmetric matrix imposes a diagonal of \(n\) 0's, multiplied by \(\sqrt{0.5}^{2}\) must be in the interval \((0,\frac{n^{2}-n}{2}]\).  

Assume without loss of generality that the extended real vector space of \(p\) variates are expressible upon a positive definite variance-covariance matrix \(\Xi_{p\times p}\). Given the definition of a positive finite variance for any non-constant random variable, along with the compact support of the distance measure \([-\frac{n^{2}-n}{2},\frac{n^{2}-n}{2}]\), we have proven that the distribution is sub-Gaussian \parencite[Ch.~1]{buldygin2000}. %By the sinusoidal transform found in equation~\ref{eq:kendall_sin}, we will show that the inverse transformation of the Spearman's \(\rho\) allows for the minimum variance expressible bound of the estimator to be expressed and numerically calculated.
The $n$ maximum values upon the Kemeny variance measure 
\begin{equation}
\label{eq:kem_variance}
\sigma_{\kappa}^{2}(X) = \sum_{k=1}^{n}\sum_{l=1}^{n} \kappa_{kl}(X)\kappa_{kl}(X) \equiv \sum_{k=1}^{n}\sum_{l=1}^{n}\kappa_{kl}^{2}(X)
\end{equation}
thus defined are realised upon the domain sub-space of a real sequence of numbers upon which is observed no ties as $a_{ij} \ne 0, \text{ for } i \ne j$. There would naturally therefore be $n!$ such occurrences, which are indeed observed. This is concordant with the claim that the maximum variance will be observed when all elements $n$ upon $x$ are observed to occur uniquely, such that no duplication occurs. Therefore, an interesting corollary is noted, wherein the Kemeny distance and correlation functions are shown to be related to the Kendall $\tau$ distance and correlation. 

% In particular, we show that these measure are linearly related, and that by redefining the minimum positive distance of the Kendall $\tau$ distance, the Kemeny distance function is obtained. A second corollary from this is also proven, wherein it is shown that the Kemeny metric basis $(X,\rho_{\kappa})$ is always denser than the Kendall $(\mathrm{S}_{n},\tau)$ metric space. 
 
 \begin{corollary}
 \label{cor:density}
 Assume that for $n>0$, the density of the permutation spaces for the respective measures are $\mathcal{M} = n^{n}$ for the Kemeny metric space, and $\mathcal{M}^{\prime} = {n}!$ for the Kendall metric space, where $\mathcal{M}^{\prime} \subset \mathcal{M}$. It would then follow that the number $\mathcal{M}$ must be well represented for the Stirling approximation of factorials \[{\sqrt {2\pi n}}\ \left(\frac {n}{e}\right)^{n}e^{\frac {1}{12n+1}} < {n}! < {\sqrt {2\pi n}}\ \left({\frac {n}{e}}\right)^{n}e^{\frac{1}{12n}},\] or else \(\mathcal{M}^{\prime} \subset \mathcal{M}\).\end{corollary}
 \begin{proof}
 Let $M \subseteq M^{\prime}$, assuming that the $\tau$-distance space is of greater or equal density to the Kemeny $\rho_{\kappa}$-space. If this were so, it would therefore follow that
 \begin{equation*}
 \begin{split}
 n^{n} & < {\sqrt {2\pi n}}\ \left({\frac {n}{e}}\right)^{n}e^{\frac {1}{12n+1}}\\
n\log(n) & < \frac{1}{2} \log(2\pi n) + n (\log(n) - \log(e)) + \big(\log(1) - \log(12n +1)\big)\log(e) < n! \\
& \hspace{1cm} <  n\log(n)  < \frac{1}{2} \log(2\pi n) + n (\log(n) - \log(e)) + \big(\log(1) - \log(12n)\big)\log(e)\\
n\log(n) & < \frac{1}{2} \log(2\pi n) + n (\log(n) - 1) + \big(0 - \log(12n +1)\big) < n! \\ 
& \hspace{1cm}<  \frac{1}{2} \log(2\pi n) + n (\log(n) - 1) + \big(0 - \log(12n)\big)\\
n\log(n) & < \frac{1}{2} \log(2\pi) + n ( \log(n) + \log(n) - 1) \equiv (\log(n)) + \big(0 - \log(12n +1)\big)\\
 n^{n} & < \sqrt{2\pi n} \cdot (\frac{n}{e})^{n}\\
 n\log(n) & < \log(2\pi n) + n\big(\log(n) - \log(e)\big)\\
 n \log(n) & < \frac{1}{2} \log\big(2\pi n\big) + n \log(n) - n\log(e)\\
 0 & < \frac{1}{2}\log\big(2\pi n) - n\\
 2n & < \log\big(2\pi n\big)\\
 \end{split}
 \end{equation*}
This conjecture is false by contradiction, as there is no $n\in \mathbb{N}^{+}$ by which $2n - \log(n) < \log(2\pi)$. It follows then that $M^{\prime} \subset M$, and therefore that all measurements upon the Kendall distance are a strict subset of the Kemeny distance, $\tau(X,Y) \subset \rho_{\kappa}(X,Y)$.  
 \end{proof}
 
 \begin{corollary}
 The linear proportionality of the two measures is found upon the subset $X$ for which $\sigma_{\kappa}^{2} = \frac{n^{2}-n}{2}$, the set of $n!$ such permutations without ties. For this set, observe that $\tau(X,Y) \in [0,\frac{n^{2}-n}{4}]$, whereas $\rho_{\kappa}(x,y) \in [0,\frac{n^{2}-n}{2}]$, and that for each permutation, $\tau(x,y) = \frac{1}{2}\rho_{\kappa}(x,y)$. \end{corollary}
 \begin{proof}
 As a direct consequence of the distance of $a$ for the $\tau$ function, wherein an adjacency swap occurs such that for $ x = \{a,b\} \to y = \{b,a\} = \tau(x,y) = 1.$ However, for $x = \{a,b\} \to y = \{b,a\} = \rho_{\kappa}(x,y) = 2,$ with the minimum non-equal distance of 1 would otherwise reflect the measurement wherein a tie occurs. Since no ties occur upon the $\tau$ distance domain, the minimum distance of a tie relative to a total ordering is undefined without loss of identification (from which the density of the respective domains is therefore consistent with Corollary~\ref{cor:density}) and therefore a distance of $\tau(x,y) = 1a \propto \rho_{\kappa} = 2a \ \therefore 2\times \tau(x,y) = \rho_{\kappa}(x,y), \forall\ (x,y) \in \mathcal{M}\ \forall\ n < \infty^{+}$. Therefore, every $\tau$-distance is validly and uniquely defined upon the $\rho_{\kappa}$-distance, but not vice-versa, as there exist no Borel \(\sigma\)-mappings for a measure space in which a tie is observed to occur upon Kendall's \(\tau\).  
 \end{proof}

We observe that there must therefore also exist a sequence of moments (as trivially guaranteed for any finite Banach norm space). As has been proven, the first expected central moment is equal to 0 as defined in equation~\ref{eq:kem_dist}, and the second moment is always finite, for all finite \(n\). This proposition is trivially validated, as the square of such a finite skew-symmetric \(n \times n\) matrix is always negative semi-definite, and therefore any sum of squares upon a totally bounded compact space is itself almost surely positive finite. The expansion to the necessity of the minimum necessary moment sequence is now discussed. A sub-Gaussian variable \(\xi\) is said to be strictly sub-Gaussian if and only if 
\begin{definition}
\label{def:stict_sg}
\begin{equation}
\label{eq:strict}
E(e^{\lambda\xi}) \le e^{\frac{\lambda^{2}\sigma^{2}_{\xi}}{2}}, \forall \lambda \in \mathbb{R} \equiv \|\xi\|_{\kappa} \le \sigma_{\xi} = \|\xi\|\ell_{2}(\Omega).
\end{equation}
\end{definition}

\begin{theorem}
The distribution of the Kemeny distance is strictly sub-Gaussian for any finite sample.
\end{theorem}

\begin{proof}
The distribution of the Kemeny distance is, for all finite \(n\), defined upon the spectrum \([-\frac{n^{2}-n}{2},\frac{n^{2}-n}{2}]\), and is therefore always finite for \(n<\infty^{+}\). The finite variance has also already been established (Lemma~\ref{lem:kem_bounded}). A sub-Gaussian variable is defined as the conjunction of a finite variance and compact totally bounded support, and therefore the Kemeny distance is sub-Gaussian for finite \(n\) \parencite[Ch.~1]{buldygin2000}. For either the bivariate or univariate case then, we observe a sub-Gaussian distribution: in the univariate scenario, the spectrum of the distance measure is \((0,b], b \in (0,\infty^{+})\)\footnote{Note that we assume \(n>0\), and therefore that \(b \ne 0\).} with expectation under affine transformation of 0, and satisfying condition~\ref{eq:strict}. For the bivariate scenario, the spectrum of the vector inner-product for the \(\kappa\) skew-symmetric matrices \(\{\kappa_{X},\kappa_{Y}\}\) is also finite and totally bounded, and is also therefore strictly sub-Gaussian for any finite \(n\). Thus, both the marginal distribution and bivariate distributions are strictly sub-Gaussian for all finite \(n\).
\end{proof}
\begin{corollary}
If then the Kemeny distance is almost surely strictly sub-Gaussian, we observe that the distribution is centred at 0, and possesses symmetric tails of density which is less than or equal to that of a standard normal distribution. Therefore it immediately follows that four moments are sufficient to characterise the probability distribution upon any population of size \(\mathcal{M}\).
\end{corollary}
\begin{proof}
It is trivially obvious then that any power of 0 is also 0, and therefore that all higher order odd-moments are equal to 0 for the Kemeny metric. The second central moment is defined in equation~\ref{eq:kem_variance}, and therefore the final free moment to examine is the excess kurtosis \(\mu_{4}\) upon a finite sample. The Kemeny distribution however is strictly sub-Gaussian and therefore possesses negative kurtosis for any finite \(n\) cannot be normally distributed. This unique probability distribution is therefore symmetric and unbiased, has a spread of scores which is almost surely positive and finite, has no skewness (by the even function property), and a finite negative excess kurtosis which tends to 0 asymptotically from below \parencite[Ch.~1]{buldygin2000}, paradoxically contradicting \textcite{diaconis1977} for finite samples, as the distribution of the finite sample distances cannot be normally distributed in the presence of ties. 
\end{proof}

%\begin{remark}
%It is also easily seen then a distribution which allows for finite sample negative excess kurtosis would guarantee a more accurate second order characterisations of the approximating functional relationship. However, while the standard errors are therefore likely to be larger than truly found when substituting the asymptotical approximation, a valid first order approximation using the likelihood function of the Gaussian probability density function is adequate; explicit comparisons of these differences may be produced by examining this relationship.  As by the generalised central limit theorem, this sum of any sub-Gaussian variable remains sub-Gaussian, however, the combination of sub-Gaussian distributions with finite support and negative excess kurtosis is potentially problematic if treated as normally distributed, however it appears to be well approximated by a stable L\'{e}vy alpha-stable distribution, for which asymptotic normality and symmetry, as well as the location and spread, may all be constructed as functions of the sample size.
%\end{remark}

\subsection{Functional analytic approximation of the Kemeny distance distribution for finite samples}
\label{subsection:function_analytic}
Typical approaches to obtaining the maximum likelihood estimator of a correlation coefficient are based upon the differentiation of a probability distribution for the respective measure space, rather than the metric function itself. We proceed to derive a number of explicit characteristics of the Kemeny metric space, including both an explicit finite sample probability distribution, and an asymptotic probability distribution. 

\begin{theorem}~\label{lem:decreasing_bounded_below}
 The sample statistics upon the Kemeny distance of size \(\mathcal{M} = n^{n}\) are monotonically convergent. %decreasing and bounded below (above) for its infinimum (supremum) in the limit upon $X$ which are less than or equal (or greater than or equal to) to the expectation for the ring upon $(X,\rho_{\kappa})^{2} \in [(\frac{-(n^{2}-n)}{2})^{2},(\frac{n^{2}-n}{2})^{2}] \equiv 0 \le \sup \frac{(n^{2}-n)^{2}}{4}$. 
 \end{theorem}
 \begin{proof}
By the greatest lower-bound property for the sequence $\{a_{m}\}_{m=1}^{(\frac{n^{2}-n}{2})^{2}}$ the finite lower-bound exists at 0, for any finite $n$. Assume that for every $\epsilon>0$, there exists any specific $N$ such that $a_{N}>(c-\epsilon)$, for arbitrary error distance \(\epsilon\) to conform to the following definition: there exists at least one measure for which $\inf_{m}\{a_{m}\} < (c-\epsilon)$, which is valid for any Banach norm-space. Then, as this set of distances upon the Kemeny metric less than or equal to the expectation is always increasing (decreasing), and for which $c$ is the lower-bound, it follows that for every population \(\mathcal{M}\) upon the finite sample \(n\) holds \[\|c-a_{n}\| \ge \|c-a_{N}\| > \epsilon \implies \lim_{m}(a_{m}^{2}) = \inf_{m}\{(a_{m})^{2}\}\ge 0.\] The distances upon \(\mathcal{M}\) are then always in the convergent interval \([0,(\frac{n^{2}-n}{2})^{2}]\), which exists for all finite \(n \in \mathbb{N}^{+}>0\).
 \end{proof}
\begin{corollary}
By negating the totally bounded and monotonically convergent spectrum of the Kemeny metric from below, we observe that the monotone convergence is confirmed to hold from above as well. By squaring the signed distances, we observe that the squared Kemeny distances are monotonically convergent from both below and above, towards the global infinimum at \(0\). Hence, the Kemeny metric is monotonically convergent from both above and below to the same point, thereby confirming the monotone convergence theorem.
\end{corollary}
\begin{remark}
We note that the monotone convergence of the Kemeny metric is uniquely minimised, for all finite \(m \in \mathcal{M}\), at the expectation of the Kemeny norm-space, which is equal to 0. Thus, the Kemeny metric also satisfies the definition of a Bregman divergence, at the point for which the variance (equation~\ref{eq:kem_variance}) is almost surely minimised within a population.
\end{remark}

The estimation of the correlation coefficient, that of the first-order characterisation, is therefore complete. However, the  distribution of \(n\) elements is only asymptotically normal, as with finite \(n\), the probability distribution for finite \(n\) of the Kemeny metric is strictly sub-Gaussian. The recognition of this property is important, as any sub-Gaussian distribution's variance will directly follow to be biased upwards for any finite \(n\) as a consequence of the coverage of of infinite support, thereby resulting in sub-optimal characterisations of the population variance. However, empirical distribution of the standard deviation of equation~\ref{eq:kem_cor} are smaller than the corresponding Kendall \(\tau_{b}\), thereby demonstrating that the sub-Gaussian distribution is majorised by the normal distribution, but for small \(n\) can be sharpened by approximately 5-13\%. 

As given by the explicit deterministic relationship between the Kemeny correlation and Spearman's \(\rho\) along with the sub-Gaussian probability distribution of the Kemeny distance function, the standard errors of the Kemeny metric should be approximately equivalent to those of Spearman's \(\rho\) transformed by equation~\ref{eq:kendall_sin}. To explore this, we produced for several sample sizes a number of correlation estimators, and examined the relative variances, which confirmed that the minimum variance property of a maximum likelihood estimator was obtained, presented in Table~\ref{tab:correlations} for 15,000 replications. Note the following definitions:

\begin{multicols}{2}
\footnotesize{
\begin{itemize}
	\item[Kemeny \(\rho\):]{Equation~\ref{eq:kemeny_rho}, equivalent to Spearman's \(\rho\).}
	\item[\(\tau\):]{Kemeny correlation (equation~\ref{eq:kem_cor}).}
	\item[\(r\):]{Pearson product-moment correlation.}
	\item[\(\rho\):]{Spearman's \(\rho\).}
	\item[\(\tau_{b}\):]{Kendall's \(\tau_{b}\) estimator \parencite{kendall1948}.}
	\item[\(\vec{r}\):]{Sinusoidal transformation of Kemeny \(\rho\).}
\end{itemize}
}%
\end{multicols}

 \begin{table}[!ht]
 \centering
 \scriptsize
 \caption{Comparison of the average error in correlation for a bivariate pair, using various methods estimated over a number of sample sizes, each time for 15,000 iterations.}
 \label{tab:correlations}
\begin{tabular}{llcccccc}
\toprule
n & Correlation  & mean & sd &  min & max & skew & kurtosis \\
\midrule
\multirow{6}{*}{n = 30} & Pearson \(r\) & -0.00262 & 0.18525 &  -0.72141 & 0.65972 & -0.01802 & -0.08994 \\
                        & Spearman \(\rho\) & -0.00281 & 0.18562 & -0.70316 & 0.66462 & -0.01807 & -0.08930 \\
                        & Kemeny \(\rho\) & -0.00281 & 0.18562  & -0.70316 & 0.66462 & -0.01807 & -0.08930 \\
                        & Kemeny \(\tau\) & -0.00200 & 0.12805  & -0.54713 & 0.48506 & -0.02348 & -0.00758 \\
                        & Kendall \(\tau_{b}\) & -0.00207 & 0.13245 & -0.56399 & 0.50238 & -0.02387 & -0.00878 \\
                        & Kemeny \(\vec{r}\) & -0.00184 & 0.12032 & -0.49645 & 0.46281 & -0.02271 & 0.05201 \\
\midrule
\multirow{6}{*}{n = 150} & Pearson \(r\) & 0.00071 & 0.08212  & -0.31894 & 0.29369 & 0.00297 & -0.01893 \\
                         		& Spearman \(\rho\) & 0.00080 & 0.08208  & -0.31567 & 0.28834 & 0.00539 & -0.01481 \\
                         		& Kemeny \(\rho\) & 0.00080 & 0.08208  & -0.31567 & 0.28834 & 0.00539 & -0.01481 \\
                         		& Kemeny \(\tau\) & 0.00051 & 0.05516  & -0.21709 & 0.20322 & 0.00844 & 0.00736 \\
                         		& Kendall \(\tau_{b}\) & 0.00052 & 0.05553  & -0.21839 & 0.20445 & 0.00843 & 0.00728 \\
                         		& Kemeny \(\vec{r}\) & 0.00051 & 0.05244  & -0.20446 & 0.18620 & 0.00577 & 0.01159 \\
\midrule
\multirow{6}{*}{n = 500} & Pearson \(r\) & 0.00012 & 0.04484  & -0.16904 & 0.17670 & -0.04287 & -0.02534 \\
                         & Spearman \(\rho\) & 0.00015 & 0.04487  & -0.16450 & 0.17300 & -0.04001 & -0.02576 \\
                         & Kemeny \(\rho\) & 0.00015 & 0.04487  & -0.16450 & 0.17300 & -0.04001 & -0.02576 \\
                         & Kemeny \(\tau\)& 0.00009 & 0.02999  & -0.11051 & 0.11766 & -0.04023 & -0.01592 \\
                         & Kendall \(\tau_{b}\) & 0.00009 & 0.03005  & -0.11074 & 0.11789 & -0.04024 & -0.01592 \\
                         & Kemeny \(\vec{r}\) & 0.00009 & 0.02860  & -0.10520 & 0.11069 & -0.04029 & -0.01830 \\
\bottomrule
\end{tabular}
 \end{table}
From the results presented in Table~\ref{tab:correlations}, we observe empirical verification of several mathematical proofs demonstrated in this work. First, the Kendall's \(\tau_{b}\) is a biased estimator of the permutation space for finite \(n\) in the presence of ties, and also does not demonstrate minimum variance properties. Thus, we can categorically determine that, as theoretically expected, the Kemeny \(\tau\) estimator developed here is a superior ML estimation function. Second, the Kemeny \(\rho\) and Spearman's \(\rho\) are equivalent, thereby allowing us to demonstrate that relative to Pearson's \(r\), these correlation estimators are both unbiased and linear functions for any linear bivariate rankings, which possesses performance greater than or equal to Pearson's \(r\), for all distributions upon a common population. Third, we observe that under sinusoidal transformation of equation~\ref{eq:kendall_sin} upon Kemeny \(\rho\) possesses tighter bounds than the Kemeny \(\tau\) estimator in equation~\ref{eq:kem_cor}. 

Further investigation as to the explicit nature of this estimator is necessary, as there appears to be a comparative gain of 5\% efficiency without the introduction of bias. At this time, we would note however that the non-linear nature of equation~\ref{eq:kendall_sin} does not invalidate the Lehmann-Scheff\'{e} condition of equation~\ref{eq:kem_cor}; instead we suspect that the comparative adjusted correction serves as an instantiation of a Rao-Blackwell style estimator, for which the unbiased Spearman's and Kemeny \(\rho\) is capable of defining a narrower continuous space than the natural permutation space of \(M\), by solving as a dual the minimisation of both the linear rank and the linear permutation distances. We hypothesise at this time that a dual characterisation of the common sample space upon two orthonormal metric spaces would be solved with a Lagrange multiplier for the equated Rayleigh quotient differences between the metric spaces. This would allow for the performance of the Kemeny \(\vec{r}\) estimation function to be linearly and uniquely solved upon both the MLE and Tikhinov regularised (i.e, biased) score vector spaces, without loss of generality. Further work will explore this relationship upon the Rayleigh quotient by equating the (potentially) regularised system of linear score equations and the system of linear rank equations.

Finally, a fourth area of mathematical development is necessary for this MLE work to continue. While the spectrum of the Kemeny metric space is strictly sub-Gaussian and asymptotically normally distributed, the operation function space upon the square vector matrix characterisation (that of \(\kappa\)) presents implementation difficulties. The utilisation of a probabilistic optimisation function as an explicit transition function is clearly a highly desirable optimisation function space, but one for which we are unable to find existing developments upon. While a vector probability estimate is directly estimable to provide maximum likelihood procedures, the convolution of a vectorised \(n \times 1\) probability distribution with a  \(n \times n\) data matrix is uncertain, as the transition across permutation matrices, i.e., the explicit linear nature of the gradient upon a vector matrix space, is currently uncertain. We suspect that this is a fundamental grounding underlying the Markov Chain relationship to continuous variable spaces.

\begin{lemma}
The Kemeny metric satisfies the strong law of large numbers for any identically and independently distributed as a linear distance function.
\end{lemma}
\begin{proof}
By Markov's inequality the existence of a finite expectation satisfies the strong law of large numbers. By the totally bounded nature of the Kemeny distance function, s.t. \(E(I\{X_{i} \le x\}) < \infty\) for all finite \(n\) then. This is equivalent to the establishing \(\sup{\sqrt{(\pm \frac{n^{2}-n}{2})^{2}}} = \frac{n^{2}-n}{2}, \,\forall 0 \le  n < \infty^{+}\). The second condition trivially holds for any finite vector sequence \(\{X_{i}\}_{i=1}^{n}\) for which a linear ordering may be performed using \(\kappa\). Thus, the strong law of large numbers is observed to hold for any finite sample. 
\end{proof}
\begin{lemma}~\label{lem:partition}
Let \(F\) be a distribution function on \(\overline{\mathbb{R}}\). For each \(\epsilon>0\) there exists a finite partition of the extended real line such that for an orderable sequence \(\infty^{-} \le t_{1} \le \cdots \le t_{k} \le \infty^{+}\), there exists \(0 \le j \le k-1\) 
\[F(t_{j+1})^{-} -  F(t_{j}) \le \epsilon.\]
\end{lemma}
\begin{proof}
Let \(0 < \epsilon\) be given, such that there exists monotone convergence. Allow \(t_{0} = \inf{\overline{\mathbb{R}}}\), for which \(j \ge 0\) we define \(t_{j+1} = \sup\{z: F(z) \le F(t_{j}) + \epsilon\}.\) Then by right continuity, there are a finite sequence of steps for which this definition is discontinuous, and we observe that for our definition of the \(\kappa\) function, this scenario does not occur upon any countable finite population. Thus is defined a transition state of monotonically decreasing distance sequences from the expectation upon \(F\).
\end{proof} 
\begin{theorem}\label{thm:gc}
The Kemeny metric function upon \(\mathcal{M}\) satisfies the Glivenko-Cantelli theorem: Let \(\{X_{i}\}_{i=1}^{\mathcal{M}}\) be an independently distributed uniform sequence of random variables with distribution function \(F \in \overline{\mathbb{R}}\). Then \[\sup_{x\in\overline{\mathbb{R}}} | \hat{F}_{m}(x) - F(x)| \to 0, a.s.\]
\end{theorem}

\begin{proof}
For any \(\epsilon>0\), holds \[\lim_{m\to\infty}\sup_{x\in\overline{\mathbb{R}}} |\hat{F}_{m}(x) - F(x)| \le \epsilon, a.s.\]  By Lemma~\ref{lem:partition} exists a partition index \(j\) for which \(t_{j} \le x < t_{j+1}\), satisficing:
\footnotesize
\begin{align*}
\hat{F}_{m}(t_{j}) \le \hat{F}_{m}(x) \le \hat{F}_{m}(t^{-}_{j+1}) \land F(t_{j}) \le F(x) \le F(t^{-}_{j+1}),\\
\implies \hat{F}_{m}(t_{j}) - F(t^{-}_{j+1}) \le \hat{F}_{m}(x) - F(x) \le \hat{F}_{m}(t^{-}_{j+1}) - F(t_{j}) \equiv \\
\hat{F}_{m}(t_{j}) - F(t_{j}) + F(t_{j}) - F(t^{-}_{j+1}) \le \hat{F}_{m}(x) - F(x) \le  \hat{F}(t^{-}_{j+1}) - F(t^{-}_{j+1}) + F(t^{-}_{j+1}) - F(t_{j})\\
\therefore \hat{F}_{m}(t_{j}) - F(t_{j}) - \frac{\epsilon}{2} \le \hat{F}_{m}(x) - F(x) \le \hat{F}_{m}(t^{-}_{j+1}) - F(t^{-}_{j+1}) + \frac{\epsilon}{2},
\end{align*}\normalsize
which tends to equality at 0 by the strong law of large numbers. Thus, the rank ordering of any extended real distribution satisfies the Glivenko-Cantelli theorem upon the Kemeny metric for any finite and therefore countable sample.
\end{proof}

\subsubsection{Haar measure}
\label{subsec:haar}
The existence of Haar measures allows us to define admissible procedures such that optimal invariant decision criteria may be established. A function as a Haar measure is defined as a unique countably additive, non-trivial measure \(\mu\) on the Borel subsets of \(G\) satisfying the following properties:
\begin{definition}~\label{def:haar}
\begin{multicols}{2}
\begin{enumerate}
    \item{The measure \(\mu\) is left-translation-invariant: \(\mu (gS)=\mu (S)\) for every \(g\in G\) and all Borel sets \(S\subseteq G\).}
    \item{The measure \(\mu\) is finite on every compact set: \(\mu (K)<\infty^{+}\) for all compact \(K\subseteq G\)}
    \item{The measure \(\mu\) is outer regular on Borel sets \(S\subseteq G\): \[\mu (S)=\inf\{\mu (U):S\subseteq U\}.\]}
    \item{The measure \(\mu\) is inner regular on open sets for compact \(K\) \(U\subseteq G:\)\[\mu (U)=\sup\{\mu (K):K\subseteq U\}.\]}
\end{enumerate}
\end{multicols}
A measure on \(G\) which satisfies these conditions is called a left Haar measure, and is a sufficient and necessary condition to establish right Haar measure existence and proportionality, and therefore equivalence. 
\end{definition}

\begin{lemma}~\label{lem:radon}
The Kemeny metric space satisfies Definition~\ref{def:radon} of a Radon measure space, thereby proving the existence of a Radon derivative.
\begin{definition}~\label{def:radon}
If \(X\) is a Hausdorff topological space, then a Radon measure on \(X\) is a Borel measure \(m\) on \(X\) such that \(m\) is locally finite and inner regular on all Borel subsets.
\end{definition}
\end{lemma}
\begin{proof}
The Kemeny metric is a Hilbert metric space (by equation~\ref{eq:kem_cor}), and thus is a \(T_{6}\), or perfectly normal space Hausdorff topological vector space. As a result, \(\kappa(X), X \in \overline{\mathbb{R}}^{n \times 1}\) must be both locally finite and inner regular. By the Riesz representation theorem, all metric spaces are inner regular on open sets \(K\). By Lemma~\ref{lem:kem_bounded}, the metric space is locally finite. Thus the finite Borel measure upon the Kemeny metric is tight (in the sense of \cite{bogachev2007} Theorem 7.1.7), and there exists a Radon derivative upon the Kemeny measure. 
\end{proof}

\begin{lemma}
\label{lem:haar}
The Kemeny metric space satisfies all properties of Definition~\ref{def:haar} and is therefore a Haar measure space.
\end{lemma}
\begin{proof}
The Kemeny metric and its Borel \(\sigma\)-algebra are closed under addition and multiplication, and therefore are left-translation-invariant. The total boundedness of Lemma~\ref{lem:kem_bounded}, guarantees the measure \(\kappa(K)\) is finite for all \(K \subseteq \overline{\mathbb{R}}\). The inner-regularity is proven Lemma~\ref{lem:radon} and the outer regularity follows as a Hilbert space. 
\end{proof}

\subsubsection{Probability distribution of the finite sample Kemeny distance}
The strict sub-Gaussian nature of the Kemeny distribution guarantees that four moments are sufficient, and that the distribution is symmetric (resulting in all odd-moments being centred at 0 w.l.g) and unbiased in expectation. Here, we approach this probability function from a discrete perspectives, using a %continuous distribution (the beta-Binomial distribution) and second a 
Beta-Binomial distribution, which unlike a normal distribution, is compact and totally bounded finite \(n\), and therefore strictly sub-Gaussian. 

\begin{lemma}~\label{lem:kem_asym_normal}
The bivariate Kemeny metric space is asymptotically normally distributed.
\end{lemma}
\begin{proof}
Consider the kurtosis of the probability distribution of the Kemeny distance function for the Beta-Binomial distribution. As a function of \(n\), it is trivially observed to be a monotonically increasing function, and by definition of strict sub-Gaussianity the excess kurtosis is always negative, here for all finite \(n\) with asymptotic normality. By the existence of all odd-moments equal to 0 and a positive variance, the monotonically convergent kurtosis tends to 3 from below as a in the asymptotic limit of \(n\). This maintains the strict sub-Gaussianity over all finite sample sizes, and is otherwise approximately normal with slight leptokurtosis (see Table~\ref{tab:1}), which becomes degenerate for the asymptotic limit.
\end{proof}

This substitution of the Beta-Binomial distribution validates by the approximate normality of the linear function space for the strict sub-Gaussianity of the Kemeny metric space. Functional approximation analysis has found the results presented in Table~\ref{tab:1}, demonstrating the change of the first, second and fourth moments relative to \(n\). Numerically, an approximation function of the variance and the excess kurtosis, as a function of \(n\), are found for \(n \ge 9\):
\begin{subequations}
\begin{equation}
\label{eq:kem_var_approx}
\hat{\sigma}^{2}_{\kappa} = 11.82 - 2.31825n +0.207355n^{2} + .110824n^{3}
\end{equation}
\begin{equation}
\label{eq:kem_kurt_approx}
% \hat{\mu}_{4} = \exp(-4.444 + 1.231n + .02793n^{2} - .00006228n^{3})
\hat{\mu}_{4} = -\exp(.0002939n^{2} - .05537n - 1.149).
\end{equation}
\end{subequations}
The restriction is negligible, as for \(n<9\), the central limit theorem is only valid asymptotically w.r.t. \(\mathcal{M}\). For this subset are found symmetric distributions where the most common frequency values are those symmetrically adjacent to the expectation (i.e., ties of distance \(\pm 1\)), whose average is then the expectation. It is also trivially confirmed, either algebraically or numerically, that the standard errors of equation~\ref{eq:kem_cor}, produced using the square root of equation~\ref{eq:kem_var_approx} scaled by \(\frac{1}{\sqrt{n-1}}\), are smaller than those traditionally produced by Kendall's \(\tau\)\footnote{\(\sigma_{\tau} \approx \sqrt {n(n-1)(2n+5)/2}\)}. This empirically confirms our initial claim and derivations of an improved and unbiased stochastically dominating test for the bivariate order independence, using the Kemeny correlation. %We note however that due to the approximate nature of the variance estimation for large \(n\), the distribution of the test statistics does follow a \(t_{n-2}\)-distribution, with \(n-2\) degrees of freedom. The relative ratio of the test statistics is roughly 3 times larger for a more precisely estimated confidence interval, confirming the minimum variance properties of a Gauss-Markov estimator function's statistical dominance as well.

\begin{table*}[!ht]
\caption{Distributional characterisation of the Kemeny distance function, for finite \(n\), permuted exhaustively for \(n \le 8\), and sampled with 3,294,172 examples for all larger \(n\).}
\label{tab:1}
\centering
\scriptsize
\begin{tabular}{lccc | lccc | lccc}
\toprule
 \(n\) & \(\mu_{1}\) & \(\sigma_{\kappa}\) & Excess \(\mu_{4}\) &  \(n\) & \(\mu_{1}\)  & \(\sigma_{\kappa}\) & Excess \(\mu_{4}\) &  \(n\) & \(\mu_{1}\) & \(\sigma_{\kappa}\) & Excess \(\mu_{4}\)\\
\midrule
 2 & 0.000 & 0.707 & -1.875 &    25 & 0.002 & 42.647 & -0.091 &   80 & -0.066 &  240.736 & -0.023\\
   3 & 0.000 & 1.610 & -1.171 &    26 & 0.019 & 45.183 & -0.083 &   85 & 0.0643 & 263.268 & -0.0274\\
   4 & 0.000 & 2.646 & -0.747 &    27 & -0.066 & 50.477 & -0.080 & 92 & 0.144 & 296.597 & -0.023\\
   5 & 0.000 & 3.795 & -0.548 &    28 & 0.040 & 53.177 & -0.076 & 96 & 0.155 & 315.766 & -0.0276\\ 
   6 & 0.000 & 5.046 & -0.432 &    29 & 0.040 & 55.900 & -0.075 & 98 & & & \\ 
   7 & 0.000 & 6.392 & -0.356 &    30 & 0.046 & 58.674 & -0.072 & 100 & -0.038 & 335.703 & -0.024 \\ 
   8 & 0.000 & 7.826 & -0.302 &    31 & 0.011 & 61.506 & -0.069 & 103 & & & \\ 
   9 & 0.000 & 9.345 & -0.259 &    32 & 0.043 & 64.418 & -0.066 & 105 & 0.0280 & 360.907 & -0.022 \\ 
   10 & 0.006 & 10.939 & -0.230 &    33 & 0.014 & 67.287 & -0.061 & 108& & & \\ 
   11 & 0.009 & 12.622 & -0.212 &    34 & 0.082 & 70.272 & -0.062 & 112 & & & \\ 
   12 & -0.007 & 14.352 & -0.191 &    35 & -0.016 & 73.262 & -0.065 & 115 & & & \\ 
   13 & -0.017 & 16.168 & -0.173 &     36 & -0.036 & 73.262 & -0.057 & 123 & & & \\ 
   14 & 0.006 & 18.064 & -0.161 &    37 & -0.005 & 76.255 & -0.063 & 125 & 0.021 & 468.456 & -0.013 \\ 
   15 & -0.010 & 19.996 & -0.148 &    38 & 0.010 & 79.419 & -0.060 & 126 & & & \\ 
   16 & -0.025 & 22.005 & -0.141 &    40 & -0.035 & 85.764 & -0.058 & 128 & & & \\  
   17 & 0.001 & 24.066 & -0.131 &    45 & -0.052 & 102.107 & -0.052 & 135 & & & \\ 
   18 & 0.010 & 26.216 & -0.122 &    50 & 0.035 & 119.342 & -0.043 & 138 & & & \\ 
   19 & -0.021 & 28.386 & -0.117 &    55 & 0.075 & 137.645 & -0.039 & 143 & & & \\ 
   20 & 0.025 & 30.645 & -0.105 &    60 & 0.057 & 156.656 & -0.036 & 147 & & & \\ 
   21 & 0.006 & 32.942 & -0.105 &    62 & -0.100 & 164.530 & -0.039 & 155 & & & \\ 
   22 & 0.008 & 35.272 & -0.096 &  64 & -0.020 & 172.447 & -0.039 & 175 & & & \\ 
   23 & 0.031 & 37.694 & -0.096 &    68 & -0.032 & 188.808 & -0.030 & 200 & & & \\ 
   24 & -0.012 & 40.155 & -0.095 &    75 & 0.042 & 218.527 & -0.026 & 225 & -0.056 & 1127.979 & -0.013 \\ 
\bottomrule
\end{tabular}
\end{table*}

% The factorial moment generating and characteristic functions for \(\mu_{1} = \frac{n^{2}-n}{2}\), derived at \(n=3\), follows:%, and is used to construct equation~\ref{eq:trunc_kemeny_var}:
% 
% \begin{subequations}
% \begin{equation}
% \label{eq:trunc_kemeny_mgf}
% \frac{1}{27}e^{-3t}\bigg(1 + 6e^{t} + 2e^{2t} + 9e^{3t} + 2e^{4t} + 6e^{5t} + e^{6t}\bigg),~ t \in \mathbb{C}
% % \frac{1}{n^{n}} e^{-nt}(1 + 2n e^{t}+3n e^{3t}+2e^{4t} + 6e^{5t} + e^{6t})
% % \frac{1}{n^{n}t^{3}} \frac{1 + 2nt + }
% % M_{\kappa(X)}(t) = \sum_{k=0}^{\infty^{+}} t^{k}\Pr(X = )
% \end{equation}
% \begin{equation}
% \label{eq:trunc_kemeny_cf}
% % \frac{1}{n^{n}}
% \frac{1}{27}(9+4\cos(t) + 12\cos(2t) + 2\cos(3t))
% \end{equation}
% \end{subequations}

The centred moments of the bivariate beta-Binomial distribution upon support \([0,n^{2}-n]\), trivially re-centred to \([\frac{n-n^{2}}{2},\frac{n^{2}-n}{2}]\) for observed sample size \(n \in \mathbb{N}\) results in the following expression of the population variance:
\begin{equation}
\label{eq:trunc_kemeny_var}
% \sigma^{2}_{\kappa} = \bigg(\frac{n^{n}-n}{n^{n}}\bigg) \Big(\frac{(n-1)^{2} \cdot (n+4)\cdot(2n-1)}{18n}\Big), ~ n \in \mathbb{N}^{+}
\sigma^{2}_{\kappa} = \frac{(n-1)^{2} \cdot (n+4)\cdot(2n-1)}{18n}, ~ n \in \mathbb{N}^{+}>1
\end{equation}

Induction upon \(n\) allows us to examine the asymptotic behaviour for the positive second and fourth moments as functions of \(\frac{n^{2}-n}{2}\), and as we previously observed, the second moment is asymptotically divergent but otherwise strongly converges for all \(n \in \mathbb{N}^{+}>2\); likewise, the fourth central moment is seen to converge to 0 from below (\(\frac{-54}{25n}\)), as given in Table~\ref{tab:1} and confirming Lemma~\ref{lem:kem_asym_normal}. The empirical approximations given in equation~\ref{eq:kem_var_approx} are also observed to be consistent with equation~\ref{eq:trunc_kemeny_var}, substantiating the empirical findings given in Table~\ref{tab:1}. 

As the distribution function is defined to be the beta-Binomial, the excess kurtosis is an explicit function of \(n\), and therefore does not require separate estimation, but is instead resolved by the convergence of the symmetric distribution to the standard normal distribution, as \(\lim_{n\to\infty^{+}}\).

\subsubsection{Cochrane's theorem upon arbitrary uniformly sampled variables}
If we accept that a beta-Binomial distribution is necessary to restrict the support to the interval of the \(\kappa\) function and thus the Kemeny metric function (see Lemma~\ref{lem:kem_bounded}), then the univariate distribution of the variance must also be a central \(\chi^{2}_{\nu=1}\) distribution, using the modified Bessel function of the first kind \(I_{v}\) is of the form:
\begin{equation}
f(X\mid \sigma_{\kappa},n) = \frac{1}{2}\bigg(\frac{X}{\sigma^{2}_{\kappa}}\bigg)^{\frac{M-1}{2}}\exp\Bigg(-\frac{X+\sigma^{2}_{\kappa}}{2}\Bigg)I_{v}(\sigma_{\kappa}\sqrt{X}),~ X = \kappa(X)
\end{equation}
To prove this, consider the set of \(\kappa^{2}\) values for any space \(\mathcal{M}\) for which the variance is positive, and for which the expectation is always 0 for all \(m\). By the Chernoff bound then, the known tails of the CDF may be obtained at the truncated points \(a = \frac{n-n^{2}}{2},b = \frac{n^{2}-n}{2}\). As the \(\kappa\) function is the basis of a complete metric space, it is closed under both addition and multiplication, and therefore allows for the sum of \(\chi^{2}\) variables to also be distributed as such. Conditions for this extension were already proven in \textcite{semrl1996}, and therefore Cochran's theorem holds for randomly sampled variables of length \(n\) upon the Kemeny measure space. From this directly follows the \(\chi^{2}_{1}\) distribution, allowing us to accept Chernoff's bound for the truncated uniform distribution without issue, purely as a function of the already proven generalised central limit theorem for strictly sub-Gaussian random variables, as is measured upon the Kemeny metric, assuming only uniform sampling independence.

A natural question then emerges, given a parametric distribution of a statistical estimator which satisfies the Gauss-Markov theorem upon any general bivariate family sampled independently and identically: how does one construct the test statistic to examine for significant differences from 0, the null hypothesis, upon a finite sample. Our representation allows us avoid the bifurcation of the variance approximation between tied and non-tied samples, and allows for an exact p-value to be uniquely determined for all finite samples, using the distances (affine linear transformations of the respective distances for the correlations) divided the standard deviations of these distributions.

% We provide our proposed test statistic upon the beta-Binomial distribution \(z_{\phi}\) upon support \[\tfrac{n^{2}-n}{2}\bigg(\tfrac{(n-1)^{2}(n+9)(2n-1)}{18n}\bigg)^{-\tfrac{1}{2}},\] 
For \(n=15\), with \(N = (n^{2}-n),\alpha_{1} = \alpha_{2}\) and \(\sigma_{\kappa}\) given by equation~\ref{eq:kem_variance}, follows a symmetric distribution centred at 0 with 2.5\% and 97.5\% quantiles of (\(-1.849937,1.850131\))\footnote{Similarly, for \(n=50\) the 95\% decision threshold for \(H_{0} = 0\) difference in mean distance is at \(z_{\Phi}(95\%) =  \pm 1.875806\).}, similarly using the asymptotic approximation of the variance for the Kendall \(\tau_{b}\) correlation estimation and test statistic \(z_{b}\):

\begin{equation}
\label{eq:z_kemeny}
z_{\phi} = \frac{-\rho_{\kappa}(X,Y)}{\sqrt{\sigma^{2}_{\kappa}}}
\end{equation}

\begin{subequations}
\begin{equation}
z_{b}=\frac{n_{c}-n_{d}}{\sqrt {v}} 
\end{equation}
\begin{equation}
\begin{array}{cl}
v     & =  \frac{(v_{0}-v_{t}-v_{u})}{18} + v_{1}+v_{2}\\
v_{0} & =  n(n-1)(2n+5)\\
v_{t} & =  \sum _{i}t_{i}(t_{i}-1)(2t_{i}+5)\\
v_{u} & =  \sum _{j}u_{j}(u_{j}-1)(2u_{j}+5)\\
v_{1} & =  \sum _{i}t_{i}(t_{i}-1)\sum _{j}u_{j}(u_{j}-1)/(2n(n-1))\\
v_{2} & =  \sum _{i}t_{i}(t_{i}-1)(t_{i}-2)\sum_{j}\frac{u_{j}(u_{j}-1)(u_{j}-2)}{(9n(n-1)(n-2))}
\end{array}
\end{equation}
\end{subequations}

\begin{figure}[!ht]
\centering
\includegraphics[scale = .45]{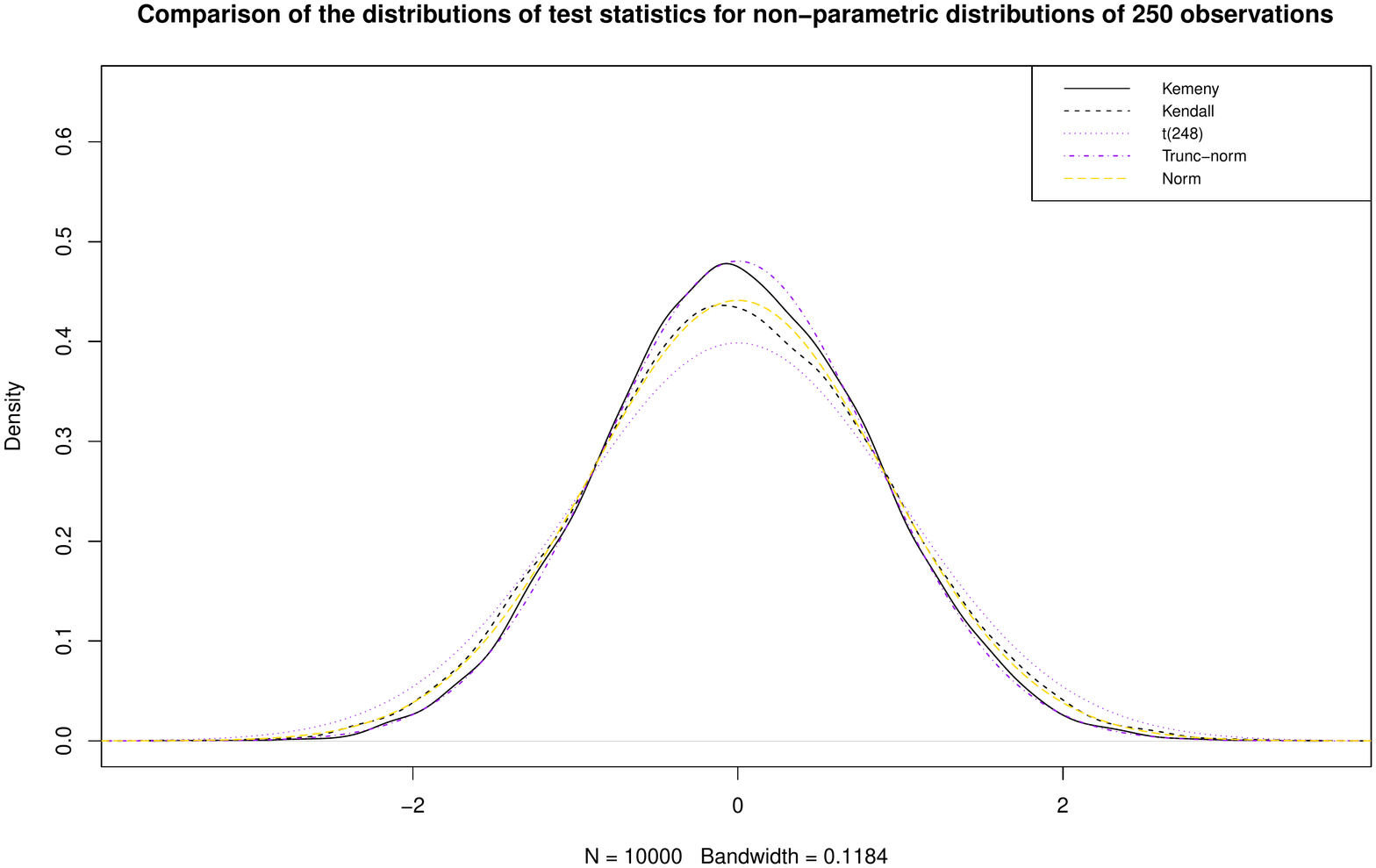}
\caption{Depiction of the distribution of \(n = 250\) replicated 10,000 times, fitted using the centred distributions of the test statistics found in Table~\ref{tab:3}.}
\label{fig:demo}
\end{figure}

As an empirical demonstration of the appropriate choice of density function and the utility of the Kemeny correlation, we construct 5,000 iterations upon multiple sample sizes of \(\{15,25,100,250,1250,2236\}\), presented Table~\ref{tab:3}. These results provide the test statistic calculated from using the respective formulas for distributions and correlation coefficients, demonstrating quite clearly that the Kendall's \(\tau_{b}\) estimator is approximately 10\% biased and simultaneously less concentrated than the corresponding Kemeny correlation in the presence of ties. 

Provided in Figure~\ref{fig:demo}, is the centred empirical density functions of the test-statistics centred at 0 for \(n = 250\), comparing the density approximations for the empirical Kemeny and Kendall \(z\) test statistics, along with: (1) the Truncated normal distribution, (2) a Student's \(t_{248}\) distribution characterising a typical Pearson \(r\) test distribution, (3) and the standardised normal distribution with standard deviation as given by \textcite{diaconis1977}. In Figure~\ref{fig:demo} for the corresponding entries from Table~\ref{tab:3}, the mean difference between the Kendall and Kemeny coefficients is not 0, and attention should be upon the shape of the tails in the distribution. The finite sample bias of the tied Kendall \(\tau_{b}\) correlation coefficient is observed to be positively and symmetrically biased away from 0, leading to an empirical confirmation that when averaging over non-measurable ties, the rate of Type I errors is inflated by approximately 10\% when using Kendall's \(\tau_{b}\), even for larger sample sizes.  

\begin{table}[!ht]
\centering
\scriptsize
\caption{Comparison of the distributions of the Kendall \(\tau_{b}\) and Kemeny correlation estimators' respective test statistic distribution, for distinct sample sizes with 5,000 replications.}
\label{tab:3}
\begin{tabular}{ccccccccc}
\toprule
                        &         & mean & sd & median & range & skew & excess kurtosis\\
\midrule
\multirow{ 2}{*}{n = 15} & Kendall & -1.39 & 0.98 & -1.47 & 6.75 & 0.33 & -0.03\\
                         & Kemeny  & -1.23 & 0.87 & -1.30 & 6.05 & 0.29 & -0.02\\
\midrule
\multirow{ 2}{*}{n = 25}  & Kendall & -1.87 & 0.94 & -1.91 & 6.65 & 0.27 & -0.09\\
                          & Kemeny  & -1.67 & 0.85 & -1.47 & 5.93 & 0.24 & -0.08\\
\midrule
\multirow{ 2}{*}{n = 100} & Kendall & -3.84 & 0.93 & -3.86 & 6.36 & 0.12 & -0.09\\
                          & Kemeny  & -3.50 & 0.85 & -3.51 & 5.91 & 0.11 & -0.08\\
\midrule
\multirow{ 2}{*}{n = 250} & Kendall & -6.07 & 0.90 & -6.08 & 6.46 & 0.11 & -0.10\\
                          & Kemeny  & -5.55 & 0.83 & -5.56 & 5.80 & 0.10 & -0.11\\
\midrule
\multirow{ 2}{*}{n = 1250} & Kendall & -13.607 & 0.972 & -13.614 & 7.625 & 0.051 & 0.021\\
                           & Kemeny  & -12.480 & 0.894 & -12.485 & 7.065 & 0.047 & 0.019\\
\midrule
\multirow{ 2}{*}{n = 2236} & Kendall & -18.207 & 0.961 & -18.218 & 7.145 & 0.030 & 0.052\\
                           & Kemeny  & -16.705 & 0.885 & -16.719 & 6.622 & 0.028 & 0.052\\
\bottomrule
\end{tabular}
\end{table}

These empirical results demonstrate two conclusive findings. First, the Kendall \(\tau_{b}\) is biased for finite samples in the presence of ties (and is a non minimum variance estimator function, as would be expected) and second, that finite sample distribution of the test statistics is uniformally most power when characterised as the beta-Binomial distribution, even for relatively large sample sizes. Thus, our estimator is demonstrated to both dominate Kendall's \(\tau_{b}\), and also to be computationally simpler as well, with exact p-values allowed subject to the assumptions of Theorem~\ref{thm:gauss-markov}.   

\subsection{Calculation of the Likelihood function}

We have obtained a finite state space upon \(\mathcal{M}\) such that under \(\lim_{m \to \mathcal{M}}\) is defined a probability distribution function \(\phi_{\kappa}\) of the bivariate correlation coefficient. This linear function space has been found to converge almost surely by the strong central limit theorem, under uniform sampling, to the beta-Binomial distribution of dimension \(k\). The natural next step is the examination of the numerical solution to maximum likelihood estimation. 

\(\mathcal{M}\) is almost surely finite, and therefore presents a sequence of viable permutation matrices \(m = \kappa(X), ~ m \in \mathcal{M},~X\in \overline{\mathbb{R}}^{n\times 1}\), for which one element is, almost surely, the most likely to be observed. This follows for a given distribution \(\bar{X} = \rho_{\kappa}(X,X) = 0, \sigma^{2}\), defined for any distance function and in equation~\ref{eq:kem_variance}, respectively. Likewise then, for any bivariate pair of independent random variables, the Hadamard product of any two such matrices which are non-degenerate characterise the permutation vector space whose solution is the suitably scaled Kemeny correlation as given in equation~\ref{eq:kem_cor}. 

The permutation space with ties upon state \(X\) as been shown to be characterised by a truncated uniform probability distribution represented by \(\kappa(X)\) with population parameters \(\sigma^{2}_{\kappa}(X)\) and the Kemeny distance from the arbitrary origin. Likewise, the bivariate combination of two such states for variables \(X,Y\) is a beta-Binomial distribution characterised by the sample size \(n\) and states, or distributions, \(\kappa(X)\) and \(\kappa(Y)\), represented by the variance-covariance matrix \(\Xi\). This permutation space is therefore representable as a Cayley graph of the data, with a finite countable set of states upon \(\mathcal{M}\) presented for the nodes \(V\) and edges \(E\): \(g = (V,E)\). The set of nodes is the set of states, \(V = \Omega\), and the set of edges \((x,y)\) is then determined by the state transition probabilities upon the Markov chain. A directed edge between two nodes therefore must represent a non-zero probability of transitioning from states \(x\to y\) with weight \(\Pr(x,y) = \phi_{\kappa}(y \mid{x})\). As a probabilistic representation, said graph must possess certain structural properties.% In particular, the weights of all outgoing (directed) edges from any node must sum to 1, the edges may be cyclical and preserve the memoryless Markov condition, thereby inducing conditional independence.

The fundamental theorem of Markov Chains asserts that if a given Markov Chain is both irreducible and aperiodic, then it is guaranteed to have a unique stationary distribution \(\pi\), and that said chain converges for any \(x \in \Omega\):
\[\lim_{t \to \infty^{+}} \Pr(x,y)^{t} \to \pi[y].\]
As would otherwise be expected, the space of \((\mathcal{M}_{x},\mathcal{M}_{y})\) and in particular, \(\mathcal{M}_{x} \times \mathcal{M}_{y}\) is quite large, given as it is by the exponential \(n^{n}-n\) with the corresponding bivariate matrix \(\Pr: n^{n}-n \times n^{n}-n\) in size, such that by sampling \(x \in \Omega\), the probability of sampling \(x\) is proportional to \(w(x)\):
\[\pi(x) = \frac{w(x)}{Z} \equiv \pi(x) \propto w(x).\] We immediately note that any sequence of scores which minimises the Kemeny distance must also immediately satisfy this requirement and thereby obtain both irreducibility and aperiodicity, and that by the Glivenko-Cantelli theorem (Theorem~\ref{thm:gc}), the transition probability \(\Pr(x\mid {y})\) is independent to the stationary distribution \(\pi\) such that: 
\begin{align*}
\Pr(x\mid y) = \Pr(x,y) \times A(x,y)\\
\pi(x) \times k(x,y) = \pi(y) \times \Pr(y,x) \times A(y,x)\\
\frac{A(x,y)}{A(y,x)} = \frac{w(y) \times \Pr(y,x)}{w(x) \times k(x,y)}
\end{align*}

% \subsubsection{Markov basis}
% A Markov basis enables the construction of a Markov chain on a function space \(f \in \mathcal{F}_{t}\). A Markov basis is defined as:
% \begin{definition}
% A set of functions \(f_{1},f_{2},\ldots,f_{L}: \mathcal{X} \to \mathbb{Z}, Z = 0, \pm 1, \pm 2, \ldots\), with the following condition (equation~\ref{eq:sum_0_basis}) holding such that for any \(t\) and \(f,f^{\prime}\in\mathcal{F}_{t}\) there are \((\epsilon_{1},f_{i_{1}}),\ldots,(\epsilon_{A},f_{i_{A}})\) with \(\epsilon = \pm 1\),
% \begin{subequations}
% \begin{equation}
% \label{eq:sum_0_basis}
% \sum_{x} f_{i}(x)^{\intercal}(x) = 0, ~1\le{i}\le{L},
% \end{equation}
% \begin{equation}
% f^{\prime} = f + \sum_{j=1}^{A}\epsilon_{i}f_{i_{j}} \land f + \sum_{j=1}^{a} \epsilon_{j=1}^{a}\epsilon_{j}f_{i_{j}} \ge 0~\text{ for } 1 \le a \le A.
% \end{equation}
% \end{subequations}
% If \(I\) is uniformly chosen upon \(f \in \mathcal{F}_{t}\) in \(1,2,\ldots,L\) and \(\epsilon = 2(a^{\intercal}a) = \pm{1}\) (for \(a\) as defined per equation~\ref{eq:kem_dist}), then the formed \(f + \epsilon\cdot{f_{I}}\) moves when non-positive, and otherwise remains.
% \end{definition}

Upon such a structure, let us define the rectangular fixed matrix \(D^{n \times p} \in \overline{\mathbb{R}}\) for which is then a vector \(\theta^{k}\) of \(k\) random variables, with distribution \(\pi(\cdot)\). Under either Bayesian or Frequentist inference then is composed the following expectation: \[E(f(\theta)) \propto \frac{\int f(x)\pi(x)\diff{x}}{\int\pi(x)\diff{x}},\] for some function \(f(\cdot)\). Let us define the transition kernel \(\Pr(x,\diff{y}) = \Pr(x,y)\diff{y} + r(x)\delta_{x}(\diff{y})\), with \(f(x): \mathcal{X} \to \mathbb{R}^{+}\) a non-negative density with equality when \(\Pr(x,x)=0\), \(\delta_{x}(\diff{y})\) is the Dirac delta function in \(\diff{y}\), and let the probability that the chain remains at \(x\) be defined as:\[r(x) = 1 - \int_{\mathcal{X}}\Pr(x,y)\diff{y}.\] Let \(\Pr(x,A)\) define the transition kernel for \(x \in \mathcal{X}\) and \(A \in \mathcal{B}(\mathcal{X})\), the Borel \(\sigma\)-field on \(\mathcal{X}\), such that \(\mathbb{P}(x,\cdot),~ \forall{x}\in\mathcal{X}\) is a probability measure and which for all \(A \in \mathcal{B}(\mathcal{X}), \Pr(\cdot,A)\) is measurable. 

We desire a transition kernel such that a unique stationary distribution results with density \(\pi(x)\) and for which the law of large numbers and the central limit theorems both apply. Then the function \[K(\theta\mid\theta^{t-1}) = \alpha(\theta\mid\theta^{t-1})q(\theta\mid\theta^{t-1}) + (1-\alpha(\theta^{t-1}))\delta_{\theta^{t-1}}(\theta),\] is a probability density in \(\theta\) for given value \(\theta^{t-1}\) measurable in \(\theta^{t-1}\), such that \(\Pr\) is a conditional probability density function expressible as \(\Pr(y\mid{x})\). This can be easily seen to be equivalent to the set of solutions \(m^{\prime} \in \mathcal{M}\) which possess both the same bivariate distance (equation~\ref{eq:kem_cor}) and univariate variances (equation~\ref{eq:kem_variance}; actually standard deviations, which by proportionality are equivalent). 

At the same time then, each cell \((i,j)~ i,j = 1,\ldots,k \in \Pr\) thereby denotes a regular conditional probability \parencite[p.~77-80]{breiman1968} constructed as the probability of \(\Pr(y\mid{x})\), subject to the restriction that \((x,y)\) are respective elements in \(\mathcal{M}\). Thus, the search is to identify the exact same condition previously described regarding the distance and marginal standard deviations, which allows a more pragmatic representation equivalence to be recognised \(\Pr(x,y) = \Pr(y\mid{x})\), over the Hadamard product of all \(n^{2}-n\) cells as per equation~\ref{eq:kem_cor}. 

Conditional independence follows for the multivariate set of independently sampled extended reals \(X^{n \times k} \in \mathcal{X}\) for \(\Pr(\theta\mid{X})\), where we re-express the kernel as the expectation of the indicator function \[E(I(X_{n})\mid X_{n-1}) = \int \Pr(x,\diff{y})I(y) = \int \delta(x_{1},\ldots,x_{i-1},y,x_{i+1},\ldots,x_{d})f(y\mid{x})\diff{y}.\] As constructed, the indicator sub-space is identical to that of the support of the Kemeny metric function space, and the Markov condition guarantees the existence of the expectation of each parameter \((x\mid{y},\theta) = (y\mid{x},\theta)\). 

Consider the approximately multivariate normal distribution for all possible functions as given upon the Kemeny metric space by \((\Xi,\phi_{\kappa})\) which includes \(f,g\) and for which \[\sup \text{cor}_{\pi}(f(x_{1}),g(x_{2}))\] is a guaranteed linear function whose convergence rate \(\lambda^{*} = \tau^{2}_{\kappa} = \arcsin(r^{2})\cdot \frac{2}{\pi}\) is bounded by the ordering induced by the Kemeny correlation \parencite{peskun1973} or the Pearson correlation (for the bivariate normal case).% As was highlighted by \textcite{roberts1992} however, the multivariate normal parameter space may not be efficiently approximated by a linear function upon the data using \(\hat{r}\): to address this, the substitution of the monotonic correlation coefficient, which is positive definite and most efficiently convergent, enables the maximum efficiency of parametrisation of the Gibbs sampler. 
\subsubsection{Maximum entropy derivation of the Kemeny correlation PMF}

For domain \([0,n^{2}-n]\) for all finite \(n\), consider the symmetric Riffled Beta Binomial distribution, treated as a mixture of two distinct Beta Binomial distributions for both even and odd distances (and therefore mutually orthogonal domains without ties and with only ties). Define the support of this distribution to be \(m = [0,1,\ldots,2m]\), such that \(E(X) = \tfrac{n^{2}-n}{2} = m\) upon the Kemeny metric, and allow the weighted mixture of the odd and even (by complement) distances \(w = \tfrac{n^{n} - n - n!}{2(n^{n} - n)} = 0.5\). 

The central moments are as follows, for \(\alpha_{1}>0,\alpha_{2}>0,w \in [0,1]\): 
\begin{subequations}
% \begin{equation}
\begin{multline}
\mu_{2} = \frac{1}{(1+2\alpha_{1}) (1+2\alpha_{2})} \Bigg(1-2m + m^{2} - w + 
2mw + 2\alpha_{2}(-1+m+w-mw+m^{2}w) - \\ 2\alpha_{1}\Big(-1 + m(2-3w) + m^{2}(w-1) + w-2\alpha_{2}(w+m-1)\Big)\Bigg) 
\end{multline}
% \end{equation}
\begin{equation}
\mu_{3} = 0
\end{equation}
\begin{multline}
\mu_{4} = 5 - 8m + 3m^{2} -5w + 6mw + \tfrac{(m-1)mw(2+3(m-1))}{2+4\alpha_{1}} - \tfrac{3mw(m-3)(m-2)(m-1)}{6+4\alpha_{1}} - \\ \tfrac{m(m-2)(m-1)(8+3(m-3))(w-1)}{2+\alpha_{2}} + \tfrac{3(w-1)(m-1)(m-2)(m-3)(m-4)}{6+\alpha_{2}}.
\end{multline}
\end{subequations}
Allowing \(\alpha_{1}>0, \alpha_{2}>0\), then holds both the general variance formula (equal to equation~\ref{eq:kem_variance}):% the latter with \(w = 0.5\):
% \begin{subequations}
\begin{equation}
\sigma^{2}_{\kappa} = \tfrac{1}{2}\bigg(\tfrac{m(m-1)}{1+2\alpha_{1}} + \tfrac{(m-1)(m-2)}{1+2\alpha_{2}} + 2m -1\bigg),
\end{equation}
% \begin{equation}
% \sigma^{2}_{\kappa} = \tfrac{1}{2}\Big(\tfrac{(m-1)(m-2)}{1+2\alpha_{2}} + \tfrac{m(m-1)}{1+2\alpha_{1}} + 2m -1\Big)
% \end{equation}
% \end{subequations}
% and kurtosis, with fixed equality \(\alpha_{1}=\alpha_{2},w=0.5\):
% \begin{subequations}
\begin{equation*}
% \gamma^{4} = %\frac{(1+2\alpha)(3((m-1)^{4} + w(-1+2m(2+m(2m-3) )) +
% 4((m-1)^{2}(3m-4) + \alpha(w(4+m(9m-11)))) + 4(5-5w+m(6w+3m-8))
% )))}{(3+2\alpha)(m^{2}+2m(\alpha+w-1) + (w-1)(2\alpha - 1))^{2}}
\frac{(2(1+\alpha)(3+6(m-1)m(2+m(m-1)) + 4\alpha(-4 + m(11+m(6m-11))) + 4\alpha^{2}(5+2m(3m-5))))}{((3+2\alpha)(1-2\alpha + 2m (2\alpha + m -1))^{2})}
\end{equation*}
where given the given definition of \(m\) may be substituted, from which follows:
\begin{equation}
% (2 (α + 1) (4 α^2 ((n^2 - n) (3/2 (n^2 - n) - 5) + 5) + 4 α (1/2 (n^2 - n) (1/2 (n^2 - n) (3 (n^2 - n) - 11) + 11) - 4) + 3 (n^2 - n) (1/2 (n^2 - n) - 1) (1/2 (n^2 - n) (1/2 (n^2 - n) - 1) + 2) + 3))/((2 α + 3) (-2 α + (n^2 - n) (2 α + 1/2 (n^2 - n) - 1) + 1)^2)
% \footnotesize{
\mu_{4} = 2\cdot\frac{\splitfrac{(2 (\alpha + 1) (4 \alpha^2 ((n^2 - n) (\tfrac{3}{2} (n^2 - n) - 5) + 5) + 4 \alpha (\tfrac{1}{2} (n^2 - n) (\tfrac{1}{2} (n^2 - n) (3 (n^2 - n) - 11) + 11) - 4) + }{3 (n^2 - n) (\tfrac{1}{2} (n^2 - n) - 1) (\tfrac{1}{2} (n^2 - n) (\tfrac{1}{2} (n^2 - n) - 1) + 2) + 3))}}{\tfrac{1}{2}((2 \alpha + 3) (-2 \alpha + (n^2 - n) (2 \alpha + \tfrac{1}{2} (n^2 - n) - 1) + 1)^2)}
% }
\end{equation}
% \end{subequations}

Given the known second moment equation we can solve for \(\alpha\) for any sample size with and without the kurtosis:
\begin{subequations}
\begin{equation}
\alpha =  -\frac{5}{2} - \frac{3\mu_{2}^{2}}{-3\mu_{2}^{2} + \mu_{4}} = \frac{9-5\mu_{4}}{2(-3+\mu_{4})},
\end{equation}
\begin{equation}
% \alpha = \frac{9n^{4} - 4n^{3} - 23n^{2}+4n-4}{8n^{3}-8n^{2}-44n + 8}
% \alpha(n) = \frac{(\frac{n^4}{2} - n^3 + n^2 + \frac{1}{6} (n - 1) \sqrt{9 n^6 - 18 n^5 - 9 n^4 + 34 n^3 + 49 n^2 - 88 n + 32} - \frac{n}{2})}{(2 (n - n^2))}
\alpha(n) = \frac{(n - 1) (9 n^3 - 4 n^2 - 14 n + 8)}{(2 (n - 2) (4 n^2 + 9 n - 4))},~ n \in \mathbb{N} > 1
\end{equation}
such that given:
\begin{equation}
q = \sqrt{2}\sqrt{\frac{\mu_{2}\mu_{4}}{3\mu_{2}^{2} - \mu_{4}}}
\end{equation}
\end{subequations}
follows the non-normalised probability mass function for arbitrary countable sample size \(n\):
\begin{equation}
\label{eq:kem_pdf}
f(x\mid \alpha,n) \propto (q^{2} - x^{2})^{\alpha}
\end{equation}
Thus completes the probability mass function of the strictly sub-Gaussian bivariate correlation. Explicitly, this section resolves the otherwise present non-measurable set upon the beta-Binomial distribution, and optimisation function problems should use this representation, or otherwise enact a Gibbs sub-sampling restriction upon the valid objective parameter solution space.

\subsubsection{Maximum likelihood derivation of the Kemeny correlation}

For \(n\) elements uniformly sampled upon \((\kappa(x),\kappa(y))\) then exists parameters \(\tau_{\kappa},\sigma^{2}_{\kappa}(x),\sigma^{2}_{\kappa}(y)\) which are distributed as \(\phi_{\kappa}\). Treating the data space as fixed, a likelihood function  then follow: \[\mathcal{L}(\tau_{\kappa},\sigma_{\kappa}(x),\sigma_{\kappa}(y)\mid \kappa(x),\kappa(y),\alpha) = \frac{\phi_{\kappa}(x,y)}{\partial{\sigma_{\kappa}(x)},\partial{\sigma_{\kappa}(y)},\partial{\tau_{\kappa}(x,y)}},\] each element of which is orthonormal, and therefore may be separably estimated. We use the notation that \(\theta = \{\sigma_{\kappa}(x),\sigma_{\kappa}(y),\tau_{\kappa}(x,y)\}\) allowing \(f = \phi_{\kappa}\) for \(\tau_{\kappa}\), and \(f = U_{a,b}(0,\sigma_{\kappa}^{2})\), as a uniform distribution with variance \(\sigma^{2}_{\kappa}\) in the interval \([0,\frac{n^{2}-n}{2}]\) denoting the order statistics.

The correlation coefficient then, allowing the variances of the correlations to be fixed to 1 without loss of generality, is the monotonically convergent product of the probabilities upon the \(\mathcal{M}_{x} \times \mathcal{M}_{y}\) state space whose score \(\frac{\partial{f}}{\partial{\theta}} = 0\), such that 
\begin{equation*}
\begin{aligned}
E \left[\frac {\partial }{\partial \theta }\log f(x,y\mid\theta,\alpha )\mid \theta \right]  & = \int _{\mathbb {R} }{\frac {{\frac {\partial }{\partial \theta }}f(x,y\mid\theta )}{f(x,y\mid\theta )}}f(x,y\mid\theta )\,\diff{\phi(x,y)}\\
 & = {\frac {\partial }{\partial \theta }}\int_{\overline{\mathbb{R}}}f(x,y\mid\theta )\,\diff{\phi(x,y)}\\
 & = {\frac {\partial }{\partial \theta }}1=0.
 \end{aligned}
\end{equation*}
This in turn allows the construction of the variance of the score, the Fisher information, as well:
\begin{equation}
\mathcal {I}(\theta )=E \left[\left({\frac {\partial }{\partial \theta }}\log f(x,y\mid\theta )\right)^{2}\mid\theta \right]=\int _{\overline{\mathbb{R}} }\left({\frac {\partial }{\partial \theta }}\log f(x,y\mid\theta )\right)^{2}f(x,y\mid\theta )\,\diff{\phi(x,y)}.
\end{equation}

Desirable properties for a maximum likelihood estimator now follows, specifically consistency and efficiency, which are proven to hold upon the Kemeny metric. Consistency is first established by examination of the Borel-Cantelli lemma for the Kemeny metric space:

\begin{lemma}
The Kemeny estimator functions and their corresponding Beta-Binomial pmf \(\phi_{\kappa}\) is a consistent estimator which satisfies the Borel-Cantelli lemma.
\end{lemma}
\begin{proof}
 The Kemeny distance has been shown to completely characterise the distribution the location of an element \(m \in \mathcal{M}\), thereby denoting the sub-additive distance from an arbitrary point of origin, and the fixed observed variance. Over the union of all such squared distance events \(E_{m}\), which have been shown to be monotonically non-increasing (Lemma~\ref{lem:decreasing_bounded_below}), continuous from above as a Haar measure (Lemma~\ref{lem:haar}), sub-additive as a Hilbert space over the extended reals (Lemma~\ref{lem:hilbert}; and therefore the compact domain Lemma~\ref{lem:kem_bounded}), and is therefore always finite upon \(m \in \mathcal{M}\), such that \(\sum_{m=1}^{\mathcal{M}}\Pr(E_{m}) < \infty^{+}\) almost surely. This follows as \(\Pr(\lim_{m\to\infty}\sup E_{m}) = 0\) for the finite collection of two extrema endpoints (thereby allowing all other points to be continuous from above) tends to 0 almost surely for finite \(m \in \mathcal{M}\): \[\lim_{m\to\infty^{+}}\lim_{n\to\infty^{+}} \frac{2}{m} \searrow \frac{2}{n^{n}-n} = 0.\]     
\end{proof}
\begin{corollary}
It is also noted that the expectation upon the Kemeny estimator parameter space is compact and therefore identified, a sufficient condition to establish consistency.
\end{corollary}
Likewise, the Cram\`{e}r-Rao lower bound is obtained upon asymptotically normal unbiased estimators upon uniformly sampled random variables.
\begin{lemma}
The Kemeny estimator functions satisfy the Cram\`{e}r-Rao lower bound upon the population \(\mathcal{M}\) constructed of asymptotic limit on \(n\).
\end{lemma}
\begin{proof}
The unbiasedness of the estimator function is established in Lemma~\ref{lem:unbiased}, and said estimator function observed to be asymptotically normally distributed as well by Lemma~\ref{lem:kem_asym_normal}. As the variance of the estimator function is strictly sub-Gaussian for all finite \(n\), the variance is a scalar constant ratio which converges to 1 as a linear function of all data distributions. Under these conditions, it follow that by the Gauss-Markov theorem (Theorem~\ref{thm:gauss-markov}) the asymptotic variance grows approach the below to the asymptotic variance of the normal distribution, and therefore concludes the proof in obtaining the Cram\`{e}r-Rao lower bound: \[\lim_{n\to\infty^{+}} \frac{n^{n}-n}{n^{n}}\sigma^{2}_{\kappa} \nearrow \sigma^{2}.\] 
\end{proof}
Verification of the lower-bound upon the Euclidean space being approached asymptotically from below for the strictly sub-Gaussian Kemeny metric space is observed in every empirical simulation conducted, wherein the variance of the Kemeny \(\tau\) correlation estimator is always smaller than the corresponding Euclidean estimator function.

\section{Examination of the Kemeny, Spearman, and Pearson correlations}
In this paper, we have examined the distribution of the null hypothesis of independence (i.e., a bivariate distance of 0) under the hypothesis of fixed independent parameters with a testing null hypothesis framework for a \(z\) statistic drawn from a beta-Binomial distribution. We have also identified a bijective relationship between Spearman's \(\rho_{S}\) and the Kemeny correlation, as defined in equation~\ref{eq:kendall_sin}. Here, we proceed to show that the testing and null hypothesis framework is identical for both the Kemeny and Spearman correlation coefficients. This allows us to thereby resolve the loss of identification in the presence of ties for both metrics with a common and identical test statistic procedure.

\subsection{Relationship between Kendall and Pearson and Polychoric correlations}
\label{subsec:kendall_sin}
In \textcite[p.~129]{kendall1948} it was claimed that 
\begin{equation}
\label{eq:kendall_sin}
\text{r}_{X,Y} = \sin\bigg({\tau_{b}(X,Y)\cdot\frac{\pi}{2}}\bigg).
\end{equation}
We proceed to show this characterisation to be invalid, and instead demonstrate that the left-hand side of the equation is actually Spearman's \(\rho\), constructable from \(\kappa(X)\) and \(\kappa(Y)\). First, consider the definition of the \(\ell_{2}\)-norm -- the insertion of one or more infinite values explicitly results in the expectation of at least one of the random variables to be non-finite. Then by the non-finite expectation, the \(n \times 1\) vector upon the Euclidean metric space is neither capable of being centred, nor is the relative distance capable of being assessed. The inner-product is also undefined, as the sum of inner-product non-finite values is itself infinite, and thus the Pearson correlation measure is degenerate. However, the Kemeny correlation is valid upon the extended reals, and produces a finite measure concomitant for any finite \(n\). This results in a paradox, wherein the Kemeny correlation exists and implies a finite convergent value, while the Pearson correlation is degenerate, which is in contradiction of equation~\ref{eq:kendall_sin}. Thus, the relationship defined in equation~\ref{eq:kendall_sin} is invalid. 

However, we can construct from the Kemeny metric and substituted for the Pearson correlation, an equivalent estimator upon the originally considered domain of \(\mathrm{S}_{n}\). Take the \(n \times n\) \(\kappa\)-permutation matrix for each extended real random variables \(\{X_{m},Y_{m}\} \in \overline{\mathbb{R}}^{n \times 1} \subset \mathcal{M}\), for which we desire a bivariate vector matrix. This is obtained by taking the sum over all \(k\) rows in the skew-symmetric permutation matrix for each random variable, resulting in the production of two \(n \times 1\) vectors, denoted as \(\vec{X}\) and \(\vec{Y}\), respectively for variables \(\{X,Y\}\): 
\begin{equation}
\label{eq:kemeny_rho}
\begin{aligned}
\vec{X}: \overline{\mathbb{R}}^{n \times 1} \to \kappa(X) \in \overline{\mathbb{R}}^{n \times n} \to  \overline{\mathbb{R}}^{n \times 1}\\
\vec{X} = \sum_{k}^{n} \kappa_{k}(X) = \bigg[ \sum_{k}^{n} \kappa_{l=1}(X) , \sum_{k}^{n}\kappa_{l=2}(X) , \cdots , \sum_{k}^{n}\kappa_{l=n}(X)\bigg]. 
\end{aligned}
\end{equation}
We note the expectations, \(E(\vec{X}) = 0\) and \(E(\vec{Y})= 0\) are fixed, and a variance (and standard deviation) measure can be equivalently defined for each univariate random variable (equation~\ref{eq:kem_variance}). Each \(\vec{X}\) is then a vector of the rank ordering of a variable, with finite mean and variance, even in the observation of a non-finite variate value in either \(X\) or \(Y\).  The inner product \(\langle\frac{\vec{X}}{\sigma_{\kappa}(X)},\frac{\vec{Y}}{\sigma_{\kappa}(Y)}\rangle\) defined in this way may then be understood to represent the angle between the two random variables. 

However, we must recognise that this inner-product is valid upon the vector of the extended reals, unlike the Pearson correlation. Upon \(\mathrm{S}_{n}\), the space of permutations wherein ties occur with probability 0, we observe that the inner-product of the ranks of finite scores and the scores themselves are equivalent: it is then only in the presence of ties in a common domain for which non-equivalent estimates of the cosine of the angle may follow. As Kendall explicitly ignored this case, the mistaken identification \(\vec{X} \equiv X\) and the respective inner-products is understandable, if still inaccurate. It is also trivially understood that both \(\vec{X}\) and \(\vec{Y}\) are the rank vectors of the original corresponding variables, and it may be observed that under Kendall's sinusoidal relationship between the inner products holds here between equations~\ref{eq:kemeny_rho} and equation~\ref{eq:kem_cor}. This invalidates the existence of a relationship between the rank and the score inner-products and instead we must replace the left hand side of equation~\ref{eq:kendall_sin} with the almost equivalent correlation upon the ranks (Spearman's \(\rho\)).

With this paradox resolved, we return to the question of the maximum likelihood estimation of the bivariate rank-ordering of a distribution of uniformly sampled variates upon \(\overline{\mathbb{R}}^{n \times p}\). This allows for linear MLE functions to be employed over non-parametric distributions, which satisfy the Cr\'{a}mer-Rao lower-bound in order to assess the sufficient statistics of any \(n \times p\) distribution. Of note, this resolves to the identification of the median and variance of each of \(p\) variates, along with the correlation matrix \(\Xi_{p \times p}\), as an \(\ell_{2}\)-norm space, thereby identifying a quadratic solution to the expectation of the median.

% \subsection{Spearman correlation}
In equation~\ref{eq:spearman_rho} it was examined, using the vectorised \(n \times 1\) bivariate variable pair upon \(\{X,Y\}\), the limiting cases for \(n = 3\), defining a space of cardinality \(|\mathcal{M}| = n^{n}-n\) elements \(m\). Here, for all \(m \in \mathcal{M}\), we observe that the vectorised permutation space is centred at 0, and contains no constant vectors, with support \([\frac{n-n^{2}}{2},\frac{n^{2}-n}{2}]\), with each extremum occurring only once upon a spanning support of \(n^{2}-n + 1\) distinct distances. With said finite support, for all finite \(n\), we observe the same necessary conditions for a strictly sub-Gaussian random variable as observed for the Kemeny metric, and also that the support of the distribution for Spearman's \(\rho_{S}\) is defined to be bijectively equivalent to that of the Kemeny correlation. Using equation~\ref{eq:spearman_rho}, we observe that an inner-product may be construct with expectation 0, and therefore produces both distance which when suitably normed is also a correlation:
\begin{equation}
\label{eq:spearman_rho}
\rho_{S} = \frac{1}{\sigma_{\kappa}(X)\sigma_{\kappa}(Y)}\tfrac{1}{n-1} \sum_{i=1}^{n} \langle \vec{X}_{(i)}\vec{Y}_{(i)}\rangle.
\end{equation}
Finally, consider that across all 24 occurrences permutations is reproduced 7 distances in tabulated frequencies and expectations which are exactly identical to the Kemeny distance, including the variance of the distances. Expressed as follows then, we construct a complete metric Frobenius space from equation~\ref{eq:spearman_rho}:
\begin{equation}
\label{eq:spearman_metric}
d_{\epsilon}(X_{(i)},Y_{(i)}) = \sqrt{2}\sqrt{1-\rho_{S}(X_{(i)},Y_{(i)})}, \rho_{S} \in [-1,1],~ X,Y \in \overline{\mathbb{R}},
\end{equation}
such that immediately follows a distance of 0 when the order statistic vectors are concordant, and a maximum distance of \(\sqrt{2}\sqrt{1-(-1)} = 2\) for the reverse image of the order statistics. We observe then that the strictly sub-Gaussian Spearman's footrule, generalised to be a quadratic signed distance function (equation~\ref{eq:spearman_rho} with fixed common sub-multiplicative variance upon the order statistics, results a test statistic possessing all minimum variance and maximum likelihood properties for the entire support with ties. 

The connection between Kendall's \(\tau_{b}\) and variations and Spearman's \(\rho\) has been long acknowledged. However, a metric topology akin to that developed for the Kemeny metric is less explicit. To resolve this, we require a Frobenius norm-space for which a mapping is defined from the domain of permutations with ties to a continuous finite function image. Using the complete metric space for random vectors \(\{X_{i},Y_{i}\}_{i=1}^{n}\), we have a viable distance measure constructed from the cross-product of the centred vector valued transformations. However, to reduce the linearity assumption, we express the Euclidean distance between two arbitrary vectors of length \(n\) which are uniformly sampled upon the bivariate extended real space. With this extension then there exists no finite first moment expectation over the population: however, if we instead focus upon the cross-product of the Kemeny vectorised representation of the order statistics, \(\{\vec{X}_{(i)},\vec{Y}_{(i)}\}_{(i=1)}^{n}\) as given in equation~\ref{eq:spearman_rho}, we obtain an asymptotically normal order-statistic space with a central moment of 0 and variance \(\sigma_{\epsilon}^{2} = n-1.\) As the Euclidean distance function is an even function, the estimator function is a valid candidate for an unbiased minimum variance estimator, and we accept that the skewness (third central moment) is almost surely 0 upon a homogeneous population. Note that by \((n-1) \sum_{i=1}^{\mathcal{M}} \langle X_{(i)},Y_{(i)}\rangle = 0,\) we trivially have established the asymptotic normality and unbiasedness of the estimating function.  

Before discussing the fourth moment, we first must consider the population however. As previously given in Subsection~\ref{subsec:kendall_sin}, the population space of Spearman's \(\rho\) is coincidental to that of the Kemeny \(\tau\) correlation estimator: in particular, it is the exhaustive linear permutation space \(\mathcal{M}\), including all ties. This allows us to immediately recognise that the asymptotically unbiased nature of the Spearman \(\rho\) correlation estimator is not normally distributed for finite samples. This follows due to the observed strict sub-Gaussianity of the estimator function, in particular the existence of a finite first and second moment over the extended real space, which is also compact and totally bounded as a function of \(n\). The distance function support is quantified by \(\sqrt{\tfrac{n^{3} - n}{3}},\)  
% \begin{subequations}
% \begin{equation}
% \frac{n^{3} - n}{3a}
% \end{equation}
% \begin{equation}
% a = 1 + \frac{1+(-1)^{n}}{2},
% \end{equation}
% \end{subequations}
presenting the square root of the sum of the squared differences over all \(n\) element rankings in \(\mathcal{M}\), and in turn the totally compact and bounded support upon the Euclidean metric space. By Definition~\ref{def:stict_sg}, we observe a permutation space which is orthonormal to that of the Kemeny metric space by equation~\ref{eq:kendall_sin}, also functionally related. 

Therefore, the excess kurtosis of the unbiased symmetric estimator of equation~\ref{eq:spearman_rho} is also monotonically convergent to 0, to satisfy the asymptotic normality of the estimator class. As a function of \(n\), the population kurtosis may be estimated by the following third order polynomial, using the reported values of Table~\ref{tab:excesskurtosis_rho}
\begin{equation}
\gamma_{2}^{2}(n) = -0.7561593 + 1.1482686n - 0.1240335n^{2} + 0.0044051n^{3}.
\end{equation}

\begin{table}
\centering
\footnotesize
\caption{\footnotesize Reported kurtosis for the Spearman's \(\rho\) correlation space for various sample sizes.}
\label{tab:excesskurtosis_rho}
\begin{tabular}{cc|cc}
\toprule
n     & Kurtosis & n & Kurtosis\\
\midrule
2 & 1     & 11 & 2.580637\\
3 & 1.5   & 12 & 2.619854\\
4 & 1.84182 & 13 & 2.643464\\
5 & 2.077129 & 14 & 2.671357\\
6 & 2.234365 & 15 & 2.695132\\
7 & 2.34464 & 16 & 2.713222\\
8 & 2.42575 & 17 & 2.728253\\
9 & 2.489407 & 18 & 2.745692\\ 
10 & 2.539668 & 19 & 2.762238\\ 
\bottomrule
\end{tabular}
\end{table}
It therefore follows that as with Kemeny's \(\tau\), the distribution of the population upon \(\mathcal{M}\) is an estimator whose distribution possesses both the minimum variance and Maximum Likelihood properties, equivalent to that of Kemeny \(\tau\). This is of course explicitly necessary, due to the pre-existing equation~\ref{eq:kendall_sin}. However, the strict sub-Gaussian nature of the random estimator variable does highlight a substantive criticism which must be noted. As the support of the distribution is compact and totally bounded, the variance is correspondingly over-estimated as conventionally estimated. Therefore, we propose an alternative null hypothesis significance testing framework for the distribution over \(\mathcal{M}\), which given the accepted standard deviation of the correlation coefficient, which may be appropriately scaled by the observed empirical standard deviation presents:
\begin{equation}
\label{eq:partial_wald_spear}
z_{\phi} = \frac{\rho_{S}}{\sqrt{n-1}}
\end{equation}
noted to be identified in the presence or absence of ties, and thus  stochastically dominates the alternative estimator functions; The beta-Binomial distribution is bounded from above and below with support \(\pm\tfrac{(n^{3}-n)(\sqrt{n-1})}{6},\) addressing the moderate under-dispersion and excess kurtosis for finite samples by accepting this distribution as being strictly sub-Gaussian. 

In line with the maximum entropy derivation of the probability distribution, we again choose the Beta-Binomial distribution, now with support \(N = \tfrac{n^{3}-n}{3}\), and 
\begin{equation}
\hspace{-3cm}
% \footnotesize{
% \alpha = \frac{\splitfrac{\tfrac{n^{2}-4n+4}{6n} + \tfrac{1}{6n\sqrt[3]{2}} ((-2n^{6}+24n^{5}-156n^{4}+608n^{3}-876n^{2}+}{\splitfrac{\sqrt{-432n^{9} + 5184n^{8} - 34560n^{7}+128736n^{6}-238464n^{5}+225936n^{4}-107136n^{3}+20736n^{2}} +}{ 528n - 128)^{\tfrac{1}{3}}) - (-(n^{2}-4+4)^{2}-12(n^{2}-n))}}}{\splitfrac{3\sqrt[3]{4}n(-2n^{6} + 24n^{5} - 156n^{4} + 608n^{3} - 876n^{2} +}{\splitfrac{ \sqrt{-432n^{9} + 5184n^{8} - 34560n^{7} + 128736n^{6} - 238464n^{5} + 225936n^{4} - 107136n^{3} + 20736n^{2}} +}{ 528n - 128)^{\tfrac{1}{3}}}}}, \, n> 0}
\alpha = \frac{\splitfrac{\tfrac{n^{2}-4n+4}{6n} + \tfrac{1}{6n\sqrt[3]{2}} ((-2n^{6}+24n^{5}-156n^{4}+608n^{3}-876n^{2}+}{\splitfrac{\sqrt{-432n^{9} + 5184n^{8} - 34560n^{7}+128736n^{6}-238464n^{5}+225936n^{4}-107136n^{3}+20736n^{2}} +}{ 528n - 128)^{\tfrac{1}{3}}) - (-(n^{2}-4+4)^{2}-12(n^{2}-n))}}}{\splitfrac{3n\sqrt[3]{4}\big(-2n^{6} + 24n^{5} - 156n^{4} + 608n^{3} - 876n^{2} + 528n - 128 + }{ \sqrt{-432n^{9} + 5184n^{8} - 34560n^{7} + 128736n^{6} - 238464n^{5} + 225936n^{4} - 107136n^{3} + 20736n^{2}}\big)^{\tfrac{1}{3}}}}, \, n> 0
% \alpha(n) = (n^4/2 - n^3 + n^2 + 1/6 (n - 1) sqrt(9 n^6 - 18 n^5 - 9 n^4 + 34 n^3 + 49 n^2 - 88 n + 32) - n/2)/(2 (n - n^2))
\end{equation}
This choice is validated with an empirical demonstration in Table~\ref{tab:5} which presents the corresponding parameter estimates and test statistics. Note that the distribution of the estimator parameters of Spearman's \(\rho_{S}\) and \(\rho_{\kappa}\) are identical in terms of absolute value skewness and kurtosis, unlike that of Kendall's and Kemeny's \(\tau\) -- however, the distribution of the parameters is more flexible when compared to that of the traditional variance approximation used under asymptotic normality, maintaining the strictly sub-Gaussian distribution through the non-positively bounded finite excess kurtosis. These results would conform to the traditional characterisation of the Spearman's \(\rho\) estimator as an unbiased estimator, as there is a bijective inverse relationship between the Spearman's footrule distance and the quadratic euclidean distance. 

Therefore, the only valid difference is determined in the evaluation of the p-values for finite samples, wherein we have observed minor, albeit optimally covering, differences, consistent with performance of an asymptotic estimator in a finite sample. Interestingly, implies that for finite samples, the use of a typical normal distribution leads to the possibility of theoretic decision errors, which are otherwise corrected for by the utilisation of the our testing procedure, due to the ill-posed dense support for the parameters. This problem was explicitly resolved by the use of the beta-binomial distribution constructed by Jayne's criterion, which also resolves the decision problems, as the support is explicitly coincidental to the Kemeny support space.

\begin{table}
\centering
\tiny
\caption{Comparison of the distributions of Spearman's \(\rho\) and Kemeny's \(\rho_{S}\) correlation estimators' respective test statistic distribution, for distinct sample sizes with 5,000 replications. Population values were re-sampled from a 2,236 bivariate sample of ordinal response items, with a population correlation of \(\rho_{S} = -0.3706851.\)}
\label{tab:5}
\begin{tabular}{ccccccccc}
\toprule
                        &         & mean & sd & median & range & skew & kurtosis\\
\midrule
\multirow{ 2}{*}{n = 15} & Spearman & 763.4964 & 141.7768  & 775.1984  & 936.1687  & -.4183  & 2.855630\\
                         & Kemeny-\(\rho_{S}\)  & -1.3597  & .9473  & -1.4379 & 6.2250  & .4183 &  2.855630\\
                         & Pearson-\(r\)  & -1.5640  & 1.4725  & -1.3958  & 12.4509 & -.8050  &  2.551283\\
\midrule
\multirow{ 2}{*}{n = 25}  & Spearman  & 3560.7549  & 489.2358  & 3585.1425  &  3497.0231 & -.3251  & 3.093031 \\
                          & Kemeny-\(\rho_{S}\)  & -1.8103  & .9218  & -1.8562  & 6.5892  & .3251  & 3.093031\\
                           & Pearson-\(r\)  & -1.9309  & 1.3617  & -1.7955  & 12.5224  & -.6580  & 4.107461\\
\midrule
\multirow{ 2}{*}{n = 100} & Spearman & 228291.6195 & 15637.7198  & 228755.8677  & 117577.3829  & -.1570  & 2.965249 \\
                          & Kemeny-\(\rho_{S}\)  & -3.6803  & .9337  & -3.7080  & 7.0200  & 0.1570  & 2.965249\\
                           & Pearson-\(r\)  & -3.6637 & 1.3057  & -3.5622  & 9.3480  & -.3269  & 2.875589\\
\midrule
\multirow{ 2}{*}{n = 250} & Spearman &  3569285.90367 & 151384.21862   &  3570740 & 1062418 & -.1051  & 3.078387\\
                          & Kemeny-\(\rho_{S}\)  & -5.839  & .9227  & -5.8684  & 6.3732 & .1051 & 2.950035\\
                           & Pearson-\(r\)   &  -5.7378 & 1.2702 & -5.7098  & 9.4235 & -0.1260  & 3.110889\\
\midrule
\multirow{ 2}{*}{n = 1250} & Spearman & 446109898.161 & 8561320.458 & 446084959 & 56305967.027 & .043 & 2.85028 \\
                           & Kemeny-\(\rho_{S}\)  & -13.092 & .929  & -13.08945 & 6.113 & -.043 & 2.85028\\
                           & Pearson-\(r\)  & -12.809 & 1.270 & -12.80797 & 10.140 & -.097 & 2.960598\\                           
\midrule
\multirow{ 2}{*}{n = 2236} & Spearman & 2553875721.27 & 36792618.22 & 2554130308 &  276930490.35 &  -.05 & 3.120799 \\
                           & Kemeny-\(\rho_{S}\)  &  -17.52 & .93  & -17.5300  & 7.03  &  .05 & 3.120799\\
                           & Pearson-\(r\)  & -17.12 & 1.27  &  -17.118 & 10.61 & -.08 & 3.190086\\
\bottomrule
\end{tabular}
\end{table}

The directed distance \(\langle\vec{X},\vec{Y}\rangle\) normed or scaled by the product of the standard deviations of the sample variances upon the population provides a standard definition for the partial Wald test statistic (equation~\ref{eq:partial_wald_spear}). As the distance is defined as a strictly sub-Gaussian random variable upon a population of finite size \(n\), the permutation block of all \(\mathcal{M}\) such elements imply a single finite homogeneous variance (or population) for all orderable (and therefore comparable) score distributions. The strict sub-Gaussianity and the necessary utilisation of the totally bounded and compact support for the beta-Binomial distribution protect against the possibility of Type I errors, and asymptotically converges to the normal distribution. Further, in the limiting case for which the Pearson correlation is nearly equivalent to Spearman's \(\rho_{S}\), we observe a well-approximated linear function space, for which the scores are a valid representation of uniform unit changes across the arbitrary domain. Therefore, while an explicit connection is made via equation~\ref{eq:kendall_sin}, we tend to observe that a system of linear equations may be characterised as a almost surely well-posed solution upon the Kemeny metric space, even if the Pearson and or Spearman variance-covariance matrices are semi-positive definite. Therefore, both the Kemeny and Euclidean distance functions are valid measurement tools with a valid and unique probability measure space, subject to viable assumptions, which may be used to solve estimation problems.       

Interestingly, with this identification problem resolved, a trivial extension exists to represent the polychoric correlation, wherein the inner-product of the latent space may be represented by the beta-binomial distribution, rather than the normal distribution, and therefore presents a resolution to the strictly sub-Gaussian but non-normal distribution otherwise resolved by the work of \textcite{browne1984} with Diagonally Weighted Least Squares. This work also resolves the finite sample properties which are otherwise only asymptotically resolved, resulting in the finite sample estimation with generalisability. 

% With the denominator fixed and identical to the Kemeny correlation coefficient \(\tau\), we now require an investigation of the numerator, to ensure that the linear bijection between the Kemeny and Euclidean metric spaces are orthonormal. Note that as constructed, the null hypothesis of the inner-product upon \(\langle\vec{X},\vec{Y}\rangle\) is 0, isometric to the distance, and may be evaluated as the distance observed upon the sample. The absolute value of the distance scaled by the population standard deviation from equation~\ref{eq:variance_spearman} thereby provides the necessary NHST framework. Given the Gauss-Markov properties of the Kemeny metric, and the corresponding inner-product constructed, we would expect the generalised central limit theorem to apply. Thus, the empirical replications upon our data should, and do, demonstrate that the variation and range of the uniformly sampled (i.e., bootstrapped) non-parametric correlation coefficients are both unbiased and demonstrate minimum variance properties against their alternatives, which are not constructed as Hilbert spaces.

% 
\subsection{The Polychoric and Kemeny correlations}

The polychoric correlation is defined in both computational and Maximum Likelihood principles for the correlation between two latent variables which are held to be bivariate normally distributed. The reasoning behind this assumption is both historical and analytic efficiency. From a historical background, Pearson largely developed the Pearson \(r\) (originally termed Galton's \(\rho\)) and effectively discounted Spearman's \(\rho\) from all consideration. From a theoretic perspective, linear functions of latent variables are only valid upon the Euclidean metric space topology under the assumption of bivariate normality, which due to the stability of the function allows the inner-product norm (i.e., the correlation) of said latent space to also be Gaussian.

However, it has been repeatedly observed that the distribution of the polychoric correlation often does not conform to a bivariate normal distribution. In fact, this contradiction resulted in the development of  \textcite{browne1984}, which sought to provide an estimator for latent variable models even if the latent variable probability space was unknown. Here, we resolve this identification problem, providing a non-parametric finite sample correlation estimator which follows the Beta-Binomial distribution. Such an estimator is asymptotically equivalent to the bivariate normal polychoric correlation, but is defined without said latent variable assumption, requiring only that the latent correlation be orderable, resulting in the Spearman's \(\rho_{S}\) correlation which results from the use of equation~\ref{eq:kendall_sin}. Given a probability distribution for said correlations it naturally of course follows that a MLE decomposition of the Spearman correlation matrix is actually possible, and thereby resolves numerous problems in the estimation of the polychoric correlation matrix. 

For a bivariate pair of random orderable (and therefore continuously linearly ranked) variables \((X,Y)\) of length \(n\) which are uniformly sampled, we observe the argument that \[r_{X,Y} = \cos\Big(\frac{\pi}{\sqrt{\frac{ad}{\frac{b}{c}}}} \Big).\] Note that however the estimation of the tetrachoric pair only relies upon the ordered set, and thus without assumption of a bivariate distribution for the pair of latent variables, \(r_{\theta_{X},\theta_{Y}}\) may be validly replaced with any inner-product norm \(\langle{X},{Y}\rangle\), including Spearman's \(\rho_{S}\) as given in equation~\ref{eq:spearman_rho}. This is immediately valid, as by the law of parallelograms for a linear metric space, the \(\cos(\theta_{X}\cdot\theta_{Y}) = \langle{X},{Y}\rangle\). 

Pearson explicitly discounted all other distributions, leaving only Pearson's \(r\) (and therefore implicitly the Euclidean distance) as the measure space, due a lack of existing alternative normed inner-product. However our estimator (equation~\ref{eq:spearman_rho}) is unbiased, and per Theorems~\ref{thm:lower_equality} and \ref{thm:upper_equality} we observe that in the presence of bivariate normality, the inner-product estimator equivalent to the standard tetrachoric correlation, by definition of the equality of the cosine of the angle and the inner-product upon the asymptotic population.

Extensions to the dimensionality of \((X,Y)\) to increase to possess more than two levels introduces the polychoric correlation. The tetrachoric correlation, by the Berry-Essen theorem, has an approximate bivariate normal distribution, linear upon the bivariate levels. However, the introduction of at least one additional response level (\(k> 2\) response levels) imposes the possibility of non-linearity directly except in the case of the Rasch model, directly introduces a non-linear but monotonic response function. This is the motivating logic underlying the two (or more) parameter item response model, and complicates the necessary minimum samples sizes for each observed response level. 

However, the logic of the polychoric correlation relies upon two sets of parameters, the bivariate correlation and the thresholds \(\delta_{f_{p}}, f = \{1,2, \ldots, k-1\}, p \in \{1,2\}\) for each bivariate variable pair assuming common lengths \(r = s = k\) for all observed variables. The estimation of the thresholds \(\delta_{f_{p}}\) is sample dependent, and introduces a linear ordering for each variable in which \(k-1\) mutually exclusive levels for variable \(p\) are defined, wherein the ordered set of \(k-1\) Heaviside functions are incorporated for each level of each variable. Note that for a non-parametric latent variable space approximated by the Heaviside functions, a linear function of the ordering upon both the Kemeny and Spearman correlations is validly produced solely as a monotonic function of the thresholds which, in expectation, are equivalent to the observed \(k-1\) levels of both \(X\) and \(Y\). From this then, we conclude that the bivariate correlation of the latent variable space, assuming monotonically non-decreasing ordering of the \(k-1\) levels upon each observed variable is equivalent to that of the bivariate Spearman's \(\rho_{S}\) correlation, and equivalently to the Kemeny \(\tau_{\kappa}\) correlation. 

From such a construction, we directly observe that the tetrachoric and, by induction upon \(k\), the polychoric correlation, as a function of estimated thresholds and the correlations between these disjoint sets, results in a ordered pair of vectors \(X_{(i)},Y_{(i)}\). However, the identification problems with the number of observations per bi-variate cell upon columns and rows \(R \times S\) for the polychoric correlation are removed here, due to the lack of a bivariate continuous distribution with the weaker orderable Beta-Binomial distribution. This directly follows from the use of a continuously linear bivariate function space directly upon the data, rather than the identification of the \(k-1\) thresholds. 

When estimated as per equation~\ref{eq:spearman_rho}, the existence of a bivariate linearly orderable distribution uniquely identified by Spearman's \(\rho_{S}\). Computationally, this resolves a dramatic concerning problem with the identification of the skewness of latent variables, as the quadratic Kemeny distance function is linear over all such distributions. The possibility of measurement error in the observed bivariate correlation introduces the subsequent latent correlation \(\rho\), which we derive here for arbitrary ordered or continuous bivariate latent distributions. The tetrachoric correlation originated as \[\Phi(r,s\mid \rho) =  \frac{1}{2\pi\sqrt{(1-\rho^{2})}} \int_{\infty^{-}}^{r}\int_{\infty^{-}}^{s} \exp\Big(- \frac{x^{2}-2\rho xy + y^{2}}{2(1-\rho^{2})}\Big)\diff{x}\diff{y},\]
and can be easily resolved for ordered set observations upon \(X,Y\) with \(R,S\) categories each, such that:
\begin{equation*}
% {
\begin{cases}
X = r~\text{if}~ a_{r} \le \theta_{1} < a_{r}, r = 1,2,\ldots,R\\
Y = s~\text{if}~ b_{s-1} \le \theta_{2} < b_{s}, s = 1,2,\ldots,S\\
\end{cases}
\end{equation*}
where \(\theta_{1},\theta_{2}\) have a bivariate beta-Binomial distribution with rank correlation coefficient \(\rho = \langle \theta_{1}\theta_{2}\rangle\). The measurement and estimation of the product-moment correlation between the latent variables, whose expectation is centred at 0 and with arbitrary variance as a function of \(n\). This is opposed to the traditional polychoric correlation, wherein the adjacent cells \(r,r+1\) and \(s,s+1\) counts are evaluated by multinomial likelihood function with normalising constant \(C\):
\begin{subequations}
\begin{equation}
L = C \prod_{r=1}^{R}\prod_{s=1}^{S} P_{rs}^{n_{ij}}
\end{equation}
\begin{equation}
P_{rs} = \Phi(a_{r},b_{s}) - \Phi(a_{r-1},b_{s}) - \Phi(a_{r},b_{s-1}) + \Phi(a_{r-1},b_{s-1}).
\end{equation}
\end{subequations}
By replacing the bivariate \(2 \times 2\) construction of the tetrachoric correlation and its bivariate Pearson's correlation with a continuous linear distribution directly over all observations, the necessity of the identification of ties is resolved. In particular, we use this to resolve the otherwise present paradox by substantially weakening the assumption upon the linear inner-product when observing observed data vectors \(X,Y\) \parencite{foldnes2019,gronneberg2020}.

Assume \(F\) is a bivariate cumulative distribution which is a bivariate beta-binomial distribution with marginal distributions \(F_{1},F_{2}\) which are uniformly distributed with variances \(\sigma^{2}_{(1)},\sigma^{2}_{(2)}\). Note that this definition, and the corresponding probability distributions, are therefore analytical solutions to application of a copula function \parencite{sklar1959}, which must exist uniquely upon the support of the order statistics with positive variance, with an upper bound in the presence of no observed ties. This assumption is typically justified for a continuous bivariate distribution, for which the probability of ties tends to 0, \textit{a.s.} by the birthday paradox, but is otherwise unnecessary for our estimator upon the Kemeny metric space, thereby providing a more adaptive linear mapping upon the domain. 

From Theorem 1 and Proposition 1 of \textcite{gronneberg2020}, we immediately note that the correlation upon the Kemeny metric is defined with a complete domain, \(\rho \in [-1,1]\) and for the set of all distributions which are linearly orderable (i.e., multinomial distributions are explicitly contained within the set of all distributions, but are not identified for the Spearman and Kemeny correlations) in a much more straightforward characterisation of the latent correlation \(\rho\). For latent variable distributions \(\eta,\xi\) which are congeneric (sharing a common latent variable), we observe construct factor loadings \(\Lambda = \vec{\lambda}^{p}\) and \(\vec{\epsilon}^{p}\) and \(\sigma_{\epsilon}^{2}\), the error variance of latent variable \(\theta:\)
\[
\begin{array}{ccc}
& X & = \lambda_{1}\theta + \epsilon_{1}\\
& Y & = \lambda_{2}\theta + \epsilon_{2}, \therefore\\
 \implies & \quad \lambda^{2}_{1} + \sigma^{2}_{\epsilon}  = \lambda^{2}_{2} + \sigma^{2}_{\epsilon}   & = 1\\
& \rho & = \lambda_{1}\lambda_{2}\\
\end{array}
\]
solved such that \(\lambda_{1} = \sqrt{\rho^{2}}, \lambda_{2} = -\lambda_{1}, \sigma_{\epsilon} = 1 - \sqrt{\rho^{2}},\) which directly follow for any distribution which is linearly orderable to be \(\rho = \rho_{S}(\lambda_{1},\lambda_{2}) = \langle\lambda_{1}\lambda_{2}\rangle,\) and allowing for the choice of sign to be arbitrary for the factor loadings. Pragmatically, given the set of distributions which observe both \((\sigma_{\epsilon}, \rho_{S})\) the empirical solution to the measurement error variance is implied both the solution to the 1 - Kemeny correlation transformed by equation~\ref{eq:kendall_sin}, thereby allowing for the orthonormal estimation of the maximum linear correlation which also allows for measurement error to be linearly identified.

Note that given the estimate of the measurement error and the factor loadings, we observe a beta-Binomial distribution for any common homogeneous population \(F\) which is linearly orderable and thereby satisfies weaker conditions than the polychoric correlation. These parameters are directly identified as a function of \(\sigma_{\epsilon}\) from which the eigendecomposition directly allows for the unique solution to the polychoric correlation. 

Numerically of course the problem is the identification of \(\sigma_{\epsilon} = \sqrt{1 - |\rho_{S}|},\) the latent correlation. Without any explicit errors observed then, \(\rho_{S} = 1\rho\), and we directly obtain \(\rho_{S} = \lambda_{1}\cdot\lambda_{2} = \sqrt{(1 - \rho)}\) for all bivariate latent correlations. \(\sigma_{\epsilon} = 1-\rho^{2}\) in turn directly depicts the algebraic solution to the problem of the identification and estimation of the latent correlation, resolved as the balance of the correlation and the error variance for both eigenvalue-eigenvector two dimensional system: \(1\), denoting the linear combination of the latent correlation, and \(2\) denoting the measurement error variance common to \((X,Y\mid\theta)\). 

Solved in this way then, we observe that both parameters are identified for all \(\rho \in [-1,1]\) and furthermore may be both numerically and algebraically solved with a just-identified, and therefore unique, solution. Upon the population the linear separability of the measurement error from the latent distribution \(\theta\) justifies the unbiasedness of the estimator, %\[\lim_{n\to\infty^{+}} \lambda_{1}\lambda_{2} + \epsilon = r(X_{(i)},Y_{(i)}), ~ \frac{E(\theta^{2})}{(\lambda_{1}\lambda_{2})^{2}} = 1,\] 
for all latent correlations which are sub-additive and therefore the product of the factor loadings between two indicators upon a common population function space. This also thereby resolves the identification of general distributions of latent variables as linear functions of the data, allowing for identified linear (generalised least squares) decomposition for both non-parametric and parametric multivariate distributions. Further, the institution of the maximum likelihood normality assumption (standard ML theory estimators) applied to the Kemeny metric conducts a valid embedding of the non-parametric space into the Euclidean metric space, proven in the following Subsection. Of course, a linear quadratic function embedding must maintain the original metric space topology: therefore, a least-squares style decomposition must result in a non-parametric latent variable space of order statistics, whose inner-product may be estimated as both the Spearman and Kemeny correlations in all cases.

\subsection{Examination of the limiting relationship between the Euclidean and Kemeny metric spaces}
We have previously suggested that the Pearson correlation and the Kemeny correlations are related, and further shown that the latter is almost surely positive definite when addressing \(p\ge 2\) random variates drawn uniformly from random distributions upon the extended reals. We will now proceed to show that, using the results of Subsection~\ref{subsec:kendall_sin}, the Kemeny and Pearson probability measures are equivalent, and therefore dual. Using this duality then, we show that while when appropriately employed, the Pearson correlation is more efficient, the largest difference of the two approximations tends to 0 almost surely, for sample sizes of at least 9. Consequently,  for ill-posed linear functions, the assumption of the Kemeny metric outperforms and uniquely and correctly solves the estimation problem more efficiently, which is achieved via an analytically tractable algorithmic procedure.

Let \(\mathcal{F}\) denote the field of measured functions upon the Euclidean metric, and let \(\mathcal{G}\) denote the corresponding measured functions upon the Kemeny metric. Further allow \(x \subseteq X\) for finite \(n \to \infty^{+}\), for which it is consequently assumed that the function space itself is \(X\) and the \(F(x) \setminus G(x) = \varnothing\), without loss of generality. By the central limit theorem, the two function spaces must converge to equality in the limit upon the population, and in turn the expectations of the two function spaces must do so as well. This last equivalence is verified by the convergence of the mean to that of the median upon the population for the respective linear function spaces, which by sub-additivity is free to move from 0 to \(\varepsilon_{0}\), the Bayes error rate, uniformly over the entire space by the Markov and Tchebyshev theoretic inequality bounds, which may be shown to satisfy the Chernoff bound as well. 

\begin{lemma}
The Euclidean metric upon the rank space (Spearman's \(\rho\), equation~\ref{eq:spearman_metric}) is a topological vector space over the space of all permutations with ties.
\end{lemma}
\begin{proof}
Let equation~\ref{eq:spearman_metric} be defined using the inner-product norm in equation~\ref{eq:spearman_rho}, for which \(\rho_{S} \in [-1,1]\). For \(\rho_{S} = 1\), the unique minima of 0 is directly obtained, and by the symmetry of the inner-product, the symmetric nature of the distance function is verified as well. Finally, sub-additivity is verified upon the Euclidean metric space for all finite measures, which includes all points upon the extended real line for the vectorised Kemeny metric space. The completeness of the distance function is now proven by the verification of its compactness. 

Consider the Identity permutation and its reverse in \(\mathcal{M}\), for which \(\rho_{S} = -1\) is uniquely identified (thereby totally bounded the convex hull of the metric space). By the Cauchy-Schwartz inequality, no distance can be greater than said extrema, and therefore the space is compact for all extended real vectors of length \(n\). As \(d_{\epsilon}\) is a compact metric space, it is also complete. Finally, \(d_{\epsilon}\) is a Hilbert space by the existence of a Banach norm space with an inner-product norm, and pre-Hilbert positive homogeneity with arbitrary scaling may be applied for the arbitrary rescaling of the vector space by constant \(a\), for which is divided out the reciprocal standard deviation scaled by \(a^{2}\), thereby maintaining a constant scaling of \(a\).
\end{proof}

\begin{theorem}
\label{thm:lower_equality}
The total variational distance between any Euclidean linear function and an unbiased Kemeny linear function is lower-bounded by 0.
\end{theorem}
\begin{proof}

Let it be accepted that the sum of random variables from both sub-Gaussian and Gaussian fields are stable, allow \(Z\) to be a random set of realised variates of length \(n\), and let \(X\) be a random variable with distribution function \(F\) and for which \(N\) is any finite number defining the compact and totally bounded support of \(F\). It then follows that the tail probability of \(X\) beyond \(N\) is the chance that \(X\) exceeds \(N\). If \(F\) is the Gaussian distribution then, the probability that \(N \ni F\), that \(N\) is not measurable upon \(F\), is 0: \[e_{G} = \Pr(|X| > N) = \Pr(X \notin [-N,N]) = F(N) - \lim_{\varepsilon \to 0^{+}} F(-N-\varepsilon).\]

However, consider now the field of \(\mathcal{G}\) which contains \(G\), and consists of the Beta-Binomial distribution for which is also measured random variable \(X\). The total variation distance between \(Y = G(Z)\) and \(X = F(Z)\) is the total variation distance between their probability distributions, where \(\mathcal{F}\) is the sigma-algebra of Borel sets upon \(F\) and \(\mathcal{G}\) is the corresponding set upon \(G\): \[\|X-Y\|_{TV} = \sup_{A\subset \mathcal{F}}\left| \Pr(X\in A) - \Pr(Y \in A) \right|.\] The Markov and Tchebyshev inequalities are immediately seen to hold for any sub-Gaussian function due to possessing finite expectations, because the expectation of \(|X|\) is bounded above by the integral, which is finite (Definition~\ref{def:stict_sg}). 

Suppose \(N\) is large enough to make \(\Pr(X \in [-N, N])\) non-zero. Truncating \(X\) at \(N\) thereby removes all chance that it exceeds \(N\), for which the value of the new distribution function at \(G(x)\) is \(0\) for \(x < -N\), \(1\) for \(x \ge N\) as defined upon \(F(x)\), and otherwise equals \[F_{[N]}(x) = \frac{F(x) - \lim_{\varepsilon \to 0^{+}} F(-N-\varepsilon)}{1-e_F(N)}.\] 

If \(F(x)\) is instead a biased function then allow \(0 < p < 1\) and \(\gamma^{2}_{M}\) be any positive number, representing the amount by which we wish to shift the mean of \(X\) quantifying the model parametrised shift in the expected asymptotic population. Should  \(X\) be any random variable with finite expectation, then allow \(\varepsilon > 0\). Pick an \(N\) for which \(e_{F}(N) \le \tfrac{\varepsilon}{2}\) and truncate \(X\) at \(N\). It then follows that the mean of \(E(X_{[N]}) = 0 + \gamma^{2}_{M}\) expressed as a \(\frac{\varepsilon}{2}\)-mixture, which changes the total variation distance by at most \(\tfrac{\varepsilon}{2}\). By sub-additivity then follows
\begin{equation}
\|X - Y\|_{TV} \le \|X - X_{[N]}\|_{TV} + \|X_{[N]} - Y\|_{TV} \le \frac{\varepsilon}{2} + \frac{\varepsilon}{2} = \varepsilon.
\end{equation}

This therefore proves that for any \(X\) has finite expectation, and therefore and independently distributed random variable within finite expectation upon the Kemeny metric, there is always a way to truncate \(X\) and shift its mean to \(\gamma^{2}_{M}\), no matter what value \(\gamma^{2}_{M}\) might have, without moving by more than \(\varepsilon\) in the total variation distance. These constructions put an upper bound on \(\|Y\|\): it is no greater than the larger of $N$ or the absolute value of the position of the atom located at \(2\tfrac{(\gamma^{2}_{M} - E(X))}{\varepsilon}\). Consequently the tails of $Y$ are zero, making them sub-Gaussian. As $\varepsilon$ may be arbitrarily small, the only possible lower bound on the distance is zero.
\end{proof}

\begin{theorem}
\label{thm:upper_equality}
The total variational distance between any Euclidean linear function and an unbiased Kemeny linear function is upper-bounded by a finite function of \(\Pr_{\kappa}(X)\), corresponding to a finite distance of \(|\tfrac{n^{2}-n}{2}|\ge{0}.\)
\end{theorem}
\begin{proof}
While the lower bound is always finite and may be treated as 0 w.l.g., for an arbitrary collection of points, the upper bound for the performance of a system \(\rho(\hat{Y},Y)\), for the arbitrary metric \(\rho\) is often indeterminate. For the Euclidean metric space, this issue is identified by the use of `approximately correct systems' which bound the measure of the space to be a finite value. With the extended real line \(\overline{\mathbb{R}}\) of performance, for which in conjunction with the Kemeny metric \(\rho_{\kappa}\) we obtain finite moments for arbitrary measure spaces, we show that the total variational distance is almost surely upper-bounded for any homogeneous function space (i.e., a function space for a common population). 

Allow \(X\) to note contain finite expectations (and therefore at least one non-finite realisation), and therefore to be a random variable. Then \(X\) is unmeasurable upon the Euclidean metric space, whereas the probability bounds are almost surely only guaranteed upon the Kemeny metric space. For any non-constant vector \(X\) then, the maximum distance between \(X\) and \(Y\), the approximation function, is \(\tfrac{n^{2}-n}{2},\) which when treated as measured and is therefore the distance error of approximation, is at most \(n^{2}-n.\) However, as \(\rho_{\kappa}\) is symmetric, the maximum distance obtainable given finite \(n\) is \(2(\tfrac{n^{2}-n}{2})\), which occurs with probability 1, as confirmed directly by the definition of the support upon the totally bounded and compact Kemeny metric space, with a Beta-Binomial distribution. 
\end{proof}

\subsubsection{Duality}

The concept of duality for projective geometry is a natural avenue of investigation. It holds that for the planar projective geometry of the Euclidean space there exists a dual permutation geometry, one which has been observed here between the Kemeny and Euclidean metric spaces. In a formal sense, by the finite nature of the Kemeny metric space, we may consider it to be a Galois field, and further a dual vector space, satisfying the three necessary properties of a dual cone. Further, the standard construction of the existence of a duality, the ability to distinguish between identical elements upon a given field with a second, is clearly self-evident upon the dual metric space characterisation. Ignoring the limiting case of the linear permutation field upon a population of linear scores (i.e., the standard asymptotic parametric learning problem per \cite{le1986}), for which a perfect parametric score fit implies an equivalent perfect ordering, denoting a bijective relationship between the ranks and the scores through the cumulative distribution function (CDF), we demonstrate that, especially in the problem of Tikhinov regularised, ill-posed or biased, learning upon linear functional map, the duality of the two metric spaces grants a just-identified unique solution.

This solution is an explicit consequence of the parametric case, in which the permutation ring of the target Kemeny distance is minimised to 0 (i.e., the correct ordering of the projective dual, the predictions, mirrors the target exactly) at the same time that the error in scores are also exactly correct. As there are a valid set of permutations which are equally distant from the target but which have different scoring error, we obtain a dual measurement which is uniquely minimised when both the Euclidean and Kemeny distances are conjointly minimised.

\subsection{Examining the Power of Rank based estimators}

As demonstrated in this manuscript, the Kemeny metric based estimators provide highly desirable properties, including unbiasedness and minimum variance with a given Beta-Binomial probability distribution for all uniformly sampled observations. These consequently have allowed us to construct standard form Neyman-Pearson Null Hypothesis Significance Testing framework procedures, with observed empirical characteristics, without either the assumptions of normality or even continuity of measures. This in turn raises the question of the power of the estimators: can we uniquely characterise the power of the rank based estimators while removing the unnecessary restrictive assumptions (namely continuity and the absence of the ties).

From \textcite{lehmann1953} we examine the null hypothesis framework, without loss of generality assuming \(H_{0}\) to be the null hypothesis of independence, and \(H_{1}\) the alternative hypothesis. Let \(F_{i} (i = 1,\ldots,n)\) be a set of continuous non-decreasing functions over the interval \([0,1]\) such that \(F_{i}(0)= 0\) and \(F_{i}(1) =1\) for which \(X_{j} = x_{1,j},x_{2,j},\ldots,x_{n,j}\), denoting the \(i^{th}\) element of the \(j^{th}\) random variable generated by the cumulative distribution function. Further, let us denote upon \(X\) the order statistics constructed upon the Kemeny metric, such that \(X_{(j)}\) denotes the order statistics of the \(j^{th}\) random variable, and element \(x_{(i,j)}\) likewise denotes the \(i^{th}\) ordered value upon the \(j^{th}\) random variable. 

Continuity for any such statistics measured upon the Kemeny metric is proven by Lemma~\ref{lem:cauchy-schwarz}, and is therefore valid upon all distributions in \(F\). Further, note that the classes \(\mathcal{F}(f_{1},\ldots,n)\) denote a partition of the family of all n-tuple described, and further, without loss of generality, may be measured upon the Kemeny metric to fall uniquely upon the discrete uniform distribution, and all bivariate pairs of random variables being continuous upon the beta-Binomial distribution. Therefore, all necessary requirements of \citeauthor{lehmann1953} are met, and therefore immediately follows a modified Lemma 3.1, Lemma 3.2 and Theorem 3.1. 

This allows for confirmation and extension of the independently derived discrete uniform distribution for each univariate distribution (albeit with a free variance parameter, due to the introduction of ties) in Lemma 3.1, and for the unique probability distribution representation, which is equivalent to that of the Riesz representation theorem applied upon the Kemeny metric (Theorem~\ref{lem:hilbert}), and immediately extends to the Neyman-Pearson fundamental Lemma \textcite{neyman1933}. We proceed to show these apply to the work of this manuscript as well.

Lemma 3.1 from \textcite[p.~25]{lehmann1953} is restated 
\begin{lemma}
\label{lem:lehman1953}
If \(F\) is a continuous cdf and if the \text{cdf} of \(Z\) is given by \(\Pr(Z \le z) = f(F(z))\) where \(f\) is non-decreasing on \([0,1]\) with \(f(0) = 0, f(1) = 1,\) then the cdf of \(F(Z) = f\).  
\end{lemma}
\begin{proof}
Lemma 3.1 holds upon the ranks as a direct consequence of the bijection of \(S_{n}\). Namely, the constancy of the distribution for each family \(\mathcal{F}\) implies the existence of \(F\) and the inherent sorting of \(X\) upon \(X_{(i,j)}\). By Corollary~\ref{cor:density} in conjunction with the application of Theorem~\ref{thm:gc}, any proofs upon the rank permutation space much uniformly hold as an inequality upon the Kemeny metric space, and thus may be treated as a worst case upper upon the Kemeny metric which may be characterised by the Beta-Binomail distribution. 
\end{proof}
\begin{corollary}
\label{cor:free}
However, the assumption of continuity to impose the lack of ties upon \(S_{n}\) is precluded, and thus is relaxed here. Therefore, the asymptotic normality (Lemma~\ref{lem:kem_asym_normal}) is still valid, yet is no longer a consequence of the continuity upon the random variables \(Z\), and therefore is removed and the constancy of the distribution is removed as well, allowing for equation~\ref{eq:kem_variance} to be a sufficient statistic upon the rank distribution, and therefore the uniqueness of the distribution of Hilbert norm distances is still constant. 
\end{corollary}

By Lemma~\ref{lem:lehman1953} and Corollary~\ref{cor:free} then held the following Theorem:
\begin{theorem}
Given any function \(f_{1}^{\circ},\ldots,f_{n}^{\circ}\) and any rank test of the hypothesis \(H: F_{1},\cdots,F_{n} \in \mathcal{F}(f_{1}^{\circ},\ldots,f_{n}^{\circ})\), the power of this test depends only on \(F_{1}:\cdots:F_{n}\). That i , if \(F_{1}:\cdots,F_{n} = F_{1}^{\prime}:\cdots,F_{n}^{\prime}\) so that \(F_{1},\ldots,F_{n}\) and \(F_{1}^{\prime},\ldots,F_{n}^{\prime}\) belong to the same class \(\mathcal{F}(f_{1},\cdots,f_{n})\) the test has the same power against these two alternatives. Furthermore, given any class of alternatives \(K: (F_{1},\cdots,F_{n}) \in \mathcal{F}(f_{1}^{\prime},\cdots,f_{n}^{\prime})\) there exists a uniformly most powerful rank testing for \(H\) against \(K\).  
\end{theorem}

In a two-sample problem, let \(X = (x_{1},x_{2},\ldots,x_{n})\) and \(Y = (y_{1},\ldots,y_{m})\), for \(m,n \in \mathbb{N}^{+}\), be independently distributed with respective cumulative distribution functions \(F_{X}\) and \(F_{Y}\), and for which is proposed the null hypothesis \(H_{0}: F = G.\) Note that upon the Kemeny metric, the bivariate distribution of the distance between \(F,G\) follows a beta-binomial distribution. As otherwise would be expected for a null hypothesis test, under the null hypothesis of equality, the expectation of the median distance over the permutation space \(\mathcal{M} \times \mathcal{M} \to [0,n^{2}-n]\) is equivalent to the point \(\tfrac{n^{2}-n}{2} \implies F(\rho_{\kappa}\mid H_{0}) = 0.5\), which is trivially confirmed. As the probability distribution of the bivariate rank distance can be now expressed in terms of the beta-Binomial distribution, we consider first the one and two-sided alternative hypotheses. 

The one-sided alternative hypothesis, \(G(x) \le F(x)\) assumes that the directed distance as given in equation~\ref{eq:kem_dist} may be measured upon the centred Kemeny distances, and thus possesses a negative distance between the two distributions. Likewise, the two-sided test is equivalent to the squared Kemeny correlation (as an affine linear transformation of the distance, the construction is bijectively acceptable), and clearly applies to the one-sided test as well. This allows us to realise the properties of the Kemeny metric in terms of the \textcite{hoeffding1948} \(D\) and obviates the \textcite{blum1961} correction. 

The generalised Hoeffding's \(D\) upon permutations with ties is constructed as follows: Let \(X_{(i)} \equiv \vec{X}\) and \(Y_{(i)} \equiv \vec{Y}\) be the order statistics constructed upon the Kemeny metric, with expectation 0 if independence is observed. The original \textcite{hoeffding1948} statistic is constructed from \(\vec{X},\vec{Y}\) using \(\phi(X_{i}) \to \{0,1\}\) as a mapping equivalent to \(\kappa^{2}(X_{i}\). The Hoeffding test of independence is explicity constructed for such a problem, using the \(\ell_{2}\)-norm upon the Euclidean probability measure space \(\diff{F_{1,2}}\), thereby requiring continuity to ensure the non-existence of ties (and therefore non-measurable events upon the symmetric group \(S_{n}\)). With the Beta-Binomial distribution for the bivariate \(F_{1,2}\) along with the marginal discrete uniform distributions for both \(F_{1}\) and \(F_{2}\), we obtain a complete measure space by which Hoeffding's independence test may be generalised to address all homogeneous populations measurable upon a common cumulative distribution function. We would note this resolution also removes the necessity of the \textcite{blum1961} coefficient, as the discontinuity does not exist. 

Note that 
\[H=\int (F_{1,2}-F_{1}F_{2})^{2}\,\diff{F_{12}},\]
for which \(\{F_{1,2},F_{1},F_{2}\}\) are respectively point-wise identified for all samples, under the assumption of a common homogeneous population. Then, \(F_{1,2} = \rho_{\kappa}(F_{1},F_{2})\sim f(F_{1},F_{2}\mid n,\alpha)\) distributed as a beta-binomial distribution, and the marginal distributions \(F_{1},F_{2}\mid \sigma_{\kappa}(F_{1}),\sigma_{\kappa}(F_{2}),\alpha =  1\) are both discrete uniform distributions, with sufficient statistics in the variances for each. Allowing for both the difference to be calculated, we obtain the Hoeffding test for independence 

% Further, allow \(\sum_{i,j} \kappa^{\interval}(X_{i})\kappa^{\interval}(Y_{i})\) as given in equation~\ref{eq:kem_cor}. 

\section{Estimating the Rank Covariance of a multivariate distribution}
A foundational problem in multivariate analysis is that of obtaining the maximum likelihood estimators of the first and second moments of a \(p\)-variate distribution based on \(X^{n \times p}, \{i = 1,\ldots,n\}, \{j = 1,\ldots,p\}\), which must therefore be assumed  to be normally distributed. Let a matrix \(A\) with \(n\) rows and \(p\) columns be denoted by \(\mathbf{A}: n \times p\), vectors \(\vec{a} = \{a_{1},\ldots,a_{n}\}\) denote row vectors, with their transposes \(\vec{a}^{\intercal}\) denoting   column vectors. \(\mathbf{D}_{a} \equiv \mathbf{D}(a_{1},\ldots,a_{n})\) is chosen to denote a diagonal matrix with diagonal elements \(a_{1},\ldots,a_{n}\). The determinant and trace are denoted \(|\mathbf{A}|\) and \(\text{tr}(\mathbf{A})\), respectively. Positive and non-negative definiteness are denoted respectively as \(\mathbf{A} > 0\) and \(\mathbf{A} \ge 0\), for which the characteristic roots \(\lambda_{1},\ldots,\lambda_{n} \) of a symmetric matrix are ordered \(\lambda_{1} \ge \cdots \ge \lambda_{n}\). Greek letters denote parameters, with upper case denoting collections thereof, typically in the form of a vector or matrix. Latin letters refer to sample values, and estimates of parameters thereof is denoted \(\theta \approx \hat{\theta}\). It also should be noted that upon the Kemeny metric then, is established a smooth quadratic \(p=2\)-norm for which the extrema is almost surely the median, thereby presenting a differentiable solution for the unknown expectation.

The positive definiteness of said matrix \(\Xi\) is trivially proven by the finite compactness of our metric space \((\mathcal{M},\rho_{\kappa})\), more generally identified as Mercer's theorem (or condition). The veracity of this allows us to characterise equation~\ref{eq:kem_cor} as a reproducing kernel Hilbert space, a necessary requirement for any Gaussian process or Support Vector Machine model. 
\begin{lemma}
The positive definiteness of the Kemeny \(\Xi\) covariance matrix holds by the almost sure positive variance of the Kemeny metric space for any finite \(n\).
\end{lemma}
\begin{proof}
Verification of this conjecture is trivially found as the matrix product of a positive diagonal matrix of the vector of standard deviations \(\mathbf{D}_{\sigma_{\kappa}}\) with correlation matrix \(\Xi\): \(\mathbf{D}\Xi\mathbf{D}\). As long as the rank of the symmetric matrix \(\Xi\) is not deficient, said matrix is almost surely positive definite.
\end{proof}

Let us assume suitably large \(n \ge 30\) such that the distribution of the Kemeny distances may be validly treated as approximately normally distributed, with even moments estimated using equations~\ref{eq:kem_var_approx} and \ref{eq:kem_kurt_approx}. For any centred distribution \(X_{j}\) which is independently distributed, the first moment is trivially both unbiased and a maximum likelihood estimator of the expectation of said vector of variates \(E(\vec{X}_{j}) = \vec{\mu}_{j}\). Trivial verification of the summations will immediately result in the verification that the expectation upon the Kemeny metric is the median score value, not the mean score value, as well.

The logarithm of the concentrated likelihood upon the Euclidean score space is defined:
\begin{subequations}
\begin{equation}
-\frac{1}{2}p n \log(2\pi) + \frac{1}{2} n\mathbf{K}
\end{equation}
\begin{equation}
\mathbf{K} = f(\Sigma;S) = - \log|\Sigma| - \text{tr}\Sigma^{-1}\mathbf{S}
\end{equation}
\begin{equation}
\mathbf{S} = \frac{1}{n} \sum_{i=1}^{n} (x_{i} - \bar{x})^{\intercal}(x_{i} - \bar{x}),
\end{equation}
\begin{equation}
\bar{x} = \frac{1}{n}\sum_{i=1}^{n} x_{i} 
\end{equation}
\end{subequations}
wherein the kernel \(K\) is defined using the population variance-covariance matrix \(\Sigma\) and sample variance-covariance matrix \(\mathbf{S}\). The problem is to maximise \(f(\Sigma;S)\) w.r.t. positive definite (p.d.) matrices \(\Sigma\), for which \(\Psi = \Sigma^{-1}\) is a one-to-one transformation of \(\Sigma\). By replacing the standard reliance upon the Euclidean \(\ell_{2}\) metric space with the Kemeny metric, we directly obtain a linear function space for the construction of the first and second moments as linear functions of the sample space, proven to be Gauss-Markov estimators (Theorem~\ref{thm:gauss-markov}). 

The explicit finite sample properties of the distribution are examined in Subsection~\ref{subsection:function_analytic}. However, if we continue with the approximation normality of the probability distribution, the resolution of maximum likelihood characterisation is immediately self-evident for suitably large sample sizes (\(n > 30\)). Under these conditions then, standard matrix differentiation methods apply upon the image of the Kemeny correlations, such that we may define the following two equivalences:
\begin{subequations}
\begin{equation}
\diff{\Xi^{-1}} \equiv -\Xi^{-1}(\diff{\Xi})\Xi^{-1}
\end{equation}
\begin{equation}
\diff{|\Xi|}_{ij} = (2 - \delta_{ij}) \Xi_{ij}(\diff{\sigma_{ij}}),
\end{equation}
\end{subequations}
for which \(\Xi_{ij}\) is accepted to be the cofactor of \(\sigma_{ij}\), and \(\delta_{ij}\) to be Kronecker's delta, following the derivations as put forth in \textcite{anderson1985}. Equivalent bounded derivations upon the function \(f(\Xi;V)\), which are monotonically convergent upon the Kemeny metric space as the function is always compact and totally bounded by its strict sub-Gaussianity, guarantees the existence of an extrema in the set of all positive definite matrices, presenting a well-posed (and thereby unique) solution for large \(n\):
\begin{subequations}
\begin{equation}
\diff[{f(\Xi;V)}]_{ij} = (2-\delta_{ij}) \bigg\{-\frac{\Xi_{ij}}{|\Xi|} \diff{\sigma_{ij}} + (\text{tr }\Xi^{-1}E_{ij}\Xi^{-1}V) \diff{\sigma_{ij}}\bigg\} = 0, i \le j,
\end{equation}
\begin{equation}
\diff[{f(\Xi;V)}]_{ij} = \Psi^{-1} - V = 0.
\end{equation}
\end{subequations}
Here \(\Xi_{ij}\) is the cofactor of \(\sigma_{ij}\) and  \(E_{ij}\) remains an identity matrix with entry \(1\) in the \((i,j)^{th}\) position and 0 in all other positions. The simplification also trivially follows with a unique solution found at \(\hat{\Psi} = V^{-1}\), which by the positive definiteness of any Gram matrix, is expressible the proportionate logarithmic convexity of the correlation matrix \(V = \mathbf{D}^{-1}\hat{\Xi}\mathbf{D}^{-1}\). The maximisation of the correlation matrix is therefore proportionate to the maximisation of the variance-covariance matrix, for any square non-singular matrix \(\mathbf{D}\), which is found for any collection of \(p\) positive standard deviations. 

Then follows the standard square matrix decomposition of the characteristic roots (i.e., the eigenvalues) which are validly applied to any positive definite Banach norm-space (which includes the Kemeny metric variance-covariance matrix). For the Rayleigh quotient then, a positive definite matrix with unit diagonal is sufficient and proportionally equivalent to the suitably scaled original variance-covariance matrix \parencite{schmidt1907}. For the bivariate normal distribution then, the correlation matrix is identical for the rank and score spaces, and almost surely positive definite for the Kemeny metric space, even in the non-linear score distribution space. This allows us to substitute the score normalisation (i.e., the inverse variance-covariance matrix) of the Kemeny correlation matrix for the Pearson variance-covariance matrix, without loss of generality, as the unique extrema exists equivalently for either metric space. 

The transformations of the diagonal scaling matrix of standard deviations (which is self-evidently the outer matrix product) thereby allows for the rank standard deviation and the standard score standard deviations to be substituted into the solution for Hadamard's inequality, such that the solution for the spread of scores is a linearly uniquely solved solution which is consistent with the Kemeny rank correlation matrix. This hypothesis is trivially proven true, as any parametric distribution must posses a uniquely defined CDF upon its score distribution, and therefore a unique linear ranking upon its scores. As the linear ranking presupposes the existence of a linear score function, its existence holds even for non-linear but monotonic score distributions. 

Consequently, the rank correlation matrix is defined for any common population of scores, but not vice-versa (and many score vector may retain the same rank ordering). As an example case, consider the estimation of the variance-covariance matrix of a bivariate pair of Cauchy distributions. In such a distribution, the expectation upon the Euclidean distance and Pearson variance-covariance matrices are undefined; however the median expectation and the bivariate variance-covariance defined upon any finite sample representation is both unbiased and convergent, and exceedingly useful occurrence as the comparative distance about the median location. This also verifies hypothesis 2.2c of \textcite{olkin1994}, in that there is indeed an exponential multivariate distribution for the order statistics, which was constructed here. 
\small
\printbibliography

 \begin{multicols}{2}
 \appendix
 \section{Definitions}
% \section{Definitions}
\begin{definition}[Metric space]
\label{def:metric_space}
A metric space is an ordered pair constructed from a non-empty set,
and a metric function on the set, which characterises a distance from the mapping of any
two elements of the set (domain) onto the range. A metric function satisfies three principal properties upon
the entire domain:
% \begin{multicols}{2}
\begin{enumerate}
  \item{\(\rho(x,y)=0\iff x=y:\)     identity of indiscernibles,}
  \item{\(\rho(x,y)=\rho(y,x):\)     symmetry, and}
  \item{\(\rho(x,z)\leq \rho(x,y)+\rho(y,z):\)     subadditivity.}
\end{enumerate}
% \end{multicols}
\end{definition}

% \begin{definition}[Cauchy sequence]
% \label{def:cauchy_seq}
% A sequence in a metric space is Cauchy if for every positive real number \(r > 0\) there is a positive integer \(N\) such that for all
% natural number sequences \(m,n \ge N, \rho(x_{m},x_{n}) < r.\)
% \end{definition}

\begin{definition}[Complete space]
\label{def:complete_space}
A metric space \((X,\rho)\) is complete if the expansion constant of the metric space is \(\le 2.\)
% any of the following
% conditions are satisfied:
  % \begin{enumerate}
    % \item{Every Cauchy sequence of points in \(X\) has a limit that is
% also in \(X\).}  \item{Every Cauchy sequence in \(X\) converges to some
% point also in \(X\) }
    % The expansion constant of \((X, \rho)\) is \(\le 2.\
    % \item{Every decreasing sequence of non-empty closed subsets of
% $X$, with diameters tending to 0, has a non-empty intersection: if
% $F_{n}$ is closed and non-empty, $F_{n+1} \subseteq F_{n}$ for every
% $n$, and $diam(F_{n}) \to 0$, then there is a point $x \in X F_{n}.$}
  % \end{enumerate}
\end{definition}
\begin{definition}[Expansion Constant]~\label{def:expansion_constant}
    The expansion constant of a metric space is the infimum of all constants $\mu$ such that whenever the family $\{{\overline {B}}(x_{\alpha },\,r_{\alpha })\}$ intersects pairwise, the intersection $\bigcap _{\alpha }{\overline {B}}(x_{\alpha },\mu r_{\alpha })$ is non-empty.
\end{definition}
\begin{definition}[Compact Metric Space]~\label{def:compact}
A metric space is compact if and only if it is complete and totally bounded.
\end{definition}

\begin{definition}[Cauchy-Schwarz inequality]~\label{def:cauchy-schwarz}
The Cauchy-Schwarz inequality states that for every pair of vectors $u$ and $v$ of an inner product space $M = n^{n}$ of fixed length $n$, holds the following equality: \[\left|\langle \mathbf {u} ,\mathbf {v} \rangle \right|^{2}\leq \langle \mathbf {u} ,\mathbf {u} \rangle \cdot \langle \mathbf {v} ,\mathbf {v} \rangle ,\] for which $\langle \cdot ,\cdot \rangle$ is the inner product. As every inner product, and hence every Hilbert space (Definition~\ref{def:hilbert}) produces the canonical norm, upon both univariate and bivariate relational pairs of vectors, $|\{\binom{p}{1},\binom{p}{2}\}| = p + \frac{p^{2}-p}{2}$, where the norm of a vector $\mathbf{u}$ is defined, $\|\mathbf {u} \|:={\sqrt {\langle \mathbf {u} ,\mathbf {u} \rangle}}$, so that this norm and the inner product are related by the defining condition $\|\mathbf {u} \|^{2}=\langle \mathbf {u} ,\mathbf {u} \rangle$, upon which $\langle \mathbf{u},\mathbf{u}\rangle$ is always a non-negative real number. It trivially follows then that taking the square root of both sides of the above inequality may be expressed as follows: \[|\langle \mathbf {u} ,\mathbf {v} \rangle |\leq \|\mathbf {u} \|\|\mathbf {v}\|.\]
% The two sides are equivalent if and only if the two vectors $\mathbf{u}$ and $\mathbf{v}$ are linearly dependent.
\end{definition}

\begin{definition}
\label{def:hilbert}
A Hilbert space is a complete metric space (Banach space) which possesses the further condition that it is positive conjugate homogeneous upon its distance function \(f\): \[f(cx) = |c|\cdot f(x),\ |c| = (c > 0).\] 
\end{definition}
  \end{multicols}
% \small
\section{Kemeny metric space as a topologically complete Hilbert space}
% In this section, we proceed to introduce a number of characteristics and proofs from the initial Axiomatic treatments of \cite{kemeny1959} for the three properties of a metric space (Definition~\ref{def:metric_space}). From these properties, we prove that the Kemeny metric is a Hilbert space and further is a Bregman divergence, upon the common domain of the Kemeny distance functional\footnote{Note that while the Kemeny metric is a functional, rather than a function, when discussing the domain, unless explicitly defined otherwise, we are referring to the domain of the interior function of the functional.}. This enables us to leverage a number of desirable traits, such as the existence of a unique \(\sigma-\)algebra (and corresponding unique probability measure) which are shown to be monotonically and uniquely convergent and to satisfy the completeness and sufficiency qualifications for their statistics.

% \subsection{The completeness of the Kemeny metric space}
Let $\rho_{\kappa}$ denote the \textbf{Kemeny metric distance function}, defined and introduced by \textcite{kemeny1959} in equation~\ref{eq:kemeny_dist1} and equation~\ref{eq:kemeny_score1}, which we axiomatically accept to denote a metric space. Allow  all capital letters to denote real vectors of length \(n\). We map any vector of length \(n\) onto a permutation matrix of order \(n\) by function \(\kappa^{*}(X_{A}): \mathbb{R}^{n} \to M^{n \times n}\), for which \(a_{ij}, i,j = 1,\ldots,n\) thereby denotes any element in said matrix in row \(i\) and column \(j\), \(a_{ij}^{*} \in M^{n \times n}\):   
\begin{subequations}
\begin{equation}
    \label{eq:kemeny_dist1}
\rho_{\kappa^{*}}(\kappa^{*}(X_{A}),\kappa^{*}(X_{B})) = \frac{1}{2} \sum_{i} \sum_{j} \text{sign}(a_{ij}^{*} - b^{*}_{ij})
\end{equation}
  \begin{equation}
    \label{eq:kemeny_score1}
    a^{*}_{ij} = {
\begin{dcases}
\:1 & \text{ if } x_{i} > x_{j}\\
-1 & \text{ if } x_{i} < x_{j}\\
\:0 & \text{ if } x_{i} = x_{i},\\
\end{dcases}
}\, i,j = 1,\ldots,n.
\end{equation}
\end{subequations}

The \(\kappa^{*}\) function represents the mapping of an extended real \(n\) length vector, such as \(X_{A}\) or \(X_{B}\), onto the respective skew-symmetric matrices \(a\) or \(b\) of order \(n \times n\). Off diagonal elements are therefore one three finite conditions defined in equation~\ref{eq:kemeny_score1}. The sign of the subtraction of corresponding entries of such two square matrices is also always finite, ensuring the summation of at most a distance \(n^{2}-n\) and a minimum of \(0\), which occur when either the data vectors are symmetrically reversed, or the data item ordering of vector \(A\) is equivalent to that of \(B\), respectively. % Of particular interest is the introduction of non-finite values, under which it is easily established that a ranking is obtained under the condition that \(\infty^{-} < a_{ij} < \infty^{+}, a_{ij} \in \mathbb{R}\) and also that \(\infty^{-} \equiv \infty^{-},\infty^{+} \equiv \infty^{+}\); in the former condition results a mapping of either \(a_{ij} \in \{-1,0,1\}\), while in the latter results a 0 in either extremum. Therefore, it may be expected that the observation of non-finite values does not preclude the construction of a finite order-centric metric space.
% Of initial focus is the function $\kappa^{*}$ which directly operates upon the domain of the metric, constructed from pairs of columns $X_{A},X_{B} \in \overline{\mathbb{R}}^{n}$. The \(\kappa^{*}\) function then directly maps an arbitrary real vectors of $n$ components onto a square matrix of order \(n \times n\), upon which the Kemeny metric distance is directly computed. We proceed to prove that the Kemeny metric space is Cauchy convergent (Definition~\ref{def:cauchy_seq}) over its entire domain.

\begin{lemma}~\label{lem:lower_bound}
The Kemeny metric space is non-negative under the condition that \(0 < a < \infty^{+}\).
% Non-negativity of the Kemeny metric space, under certain conditions, follows from the axiomatic use of the metric properties (Definition~\ref{def:metric_space}) as proven by \cite{kemeny1959}. 
When the domain of a univariate random variable in \(\mathbb{R}^{n \times 1}\) is extended to a bivariate pair \((X_{A},X_{B})\), the function \(\kappa(X)\) is convex, for the arbitrary arbitrary real constant \(\alpha\).
\end{lemma}
\begin{proof}
By sequential application of sub-additivity, symmetry, and the identity of indiscernibles follows non-negativity, but only under a specific definition of $\kappa(\cdot)$, upon which results the requirement $0 < a < \infty^{+}$:
\begin{equation}
\begin{aligned}
% \rho_{\kappa^{*}}(\kappa^{*}(X_{A}),\kappa^{*}(X_{B})) + \rho_{\kappa^{*}}(\kappa^{*}(X_{B}),\kappa^{*}(X_{A})) & \ge \rho_{\kappa^{*}}(\kappa^{*}(X_{A}),\kappa^{*}(X_{A}))\\
% \rho_{\kappa^{*}}(\kappa^{*}(X_{A}),\kappa^{*}(X_{B})) + \rho_{\kappa^{*}}(\kappa^{*}(x),\kappa^{*}(X_{B})) & \geq \rho_{\kappa^{*}}(\kappa^{*}(X_{A}),\kappa^{*}(X_{A}))  \\
% 2\cdot \rho_{\kappa^{*}}(\kappa^{*}(X_{A}),\kappa^{*}(X_{B})) & \ge 0\\
\rho_{\kappa^{*}}(a,b) + \rho_{\kappa^{*}}(b,a) & \ge \rho_{\kappa^{*}}(a,a)\\
\rho_{\kappa^{*}}(a,b) + \rho_{\kappa^{*}}(a,b) & \geq \rho_{\kappa^{*}}(a,b)  \\
2\cdot \rho_{\kappa^{*}}(a,b) & \equiv 2\cdot \rho_{\kappa^{*}}(\kappa^{*}(X_{A},\kappa^{*}(X_{B})) \ge 0\\
\rho_{\kappa^{*}}(\kappa^{*}(X_{A},\kappa^{*}(X_{B})) \ge 0\\
\end{aligned}
\end{equation}
\end{proof}

\begin{lemma}~\label{lem:kem_expansion}
% An alternative proof of the completeness of the Kemeny metric space is also offered by the expansion constant of the Kemeny distance function, from Definition~\ref{def:expansion_constant}, which is here shown to be less than 2 for the Kemeny metric for finite \(n\):
The Kemeny distance function has an expansion constant \(\mu\) less than 2.
\end{lemma}
\begin{proof}
Let there exist two families $\{A_{i}\}_{i \in I}$ and $\{B_{j}\}_{j \in J\,, i,j = 1,\ldots,n}$ of $n$ real numbers such that $a_{i} \le b_{j}$ and therefore \begin{equation}
\begin{aligned}
\rho_{\kappa^{*}}(\kappa^{*}(A_{i}),\kappa^{*}(B_{j}))\, \forall\, i,j \implies & \sup_{i\in I}(\kappa^{*}(A_{i}))  \equiv \mu \\
& = \sup\{-1\cdot{a},0\cdot{a}\} \le \inf_{j\in J}(\kappa^{*}(B_{j})) \equiv \inf\{0\cdot{a},1\cdot{a}\}\\
& = \sup\{\mu\} \le 1.
\end{aligned}
\end{equation}
It therefore follows that for any family of intervals $A_{i},B_{j}$ whose intersection of the two finite sets is never empty, $\cap[\sup{A_{i}},\inf{B_{j}}] \ne \varnothing$, possesses an expansion constant $\mu = \sup\{0,1a\} \le 2a\; \forall a>0$. Consequently, Definition~\ref{def:complete_space} is satisfied for the Kemeny distance function, in that the isometric expansion constant is always less than 2; as \(X_{A} = A_{i},X_{B}=B_{j}\), this proof applies to all score vectors within \(\mathcal{X}\), the space of all vectors of extended reals of finite length \(n\).
\end{proof}
\begin{lemma}~\label{thm:convexity}
% Convexity of the $\kappa(X)$ function image upon the domain $X \in \mathcal{X}$, from which follows also the extension of the domain of a univariate random variable $x \in \overline{\mathbb{R}}^{n\times 1}$ to any bivariate pair $(X_{A},X_{B}) \in \mathcal{X}^{n}$ is proven.
When the domain of a univariate random variable in \(\overline{\mathbb{R}}^{n \times 1}\) is extended to a bivariate pair \((X_{A},X_{B})\), the function \(\kappa^{*}(X)\) is convex.
\end{lemma}

\begin{proof}
% Recall that for any norm there exists a function \(f=\rho_{\kappa}(\kappa(X_{A}),\kappa(X_{B})) \lor \kappa(X_{A})\kappa(X_{B})\), for which we allow \(X_{B} = I_{n}\) without loss of generality, such that is satisfied all following properties: 
Recall that for any norm there exists a function \(f=\rho_{\kappa^{*}}(\kappa^{*}(X_{A}),\kappa^{*}(X_{B})) \equiv \rho_{\kappa^{*}}\). If we allow without loss of generality \(X_{B} = I_{n}\), the Identity permutation, such that is satisfied all following properties: 
\begin{gather}
\begin{aligned}
% \rho_{\kappa^{*}}(\lambda \kappa^{*}(X_{A}),\kappa^{*}(X_{B})) = |a|\rho_{\kappa^{*}}(\kappa^{*}(X_{A}),\kappa^{*}(X_{B})),\forall \alpha \in \overline{\mathbb{R}},\\
\rho_{\kappa^{*}}(\kappa^{*}(X_{A}),  \kappa^{*}(X_{B})) \le \rho_{\kappa^{*}}(\kappa^{*}(X_{A}),\kappa^{*}(X_{B})) + \rho_{\kappa^{*}}(\kappa^{*}(X_{B}),\kappa^{*}(X_{B})),\\
\rho_{\kappa^{*}}(\kappa^{*}(X_{A}),\kappa^{*}(X_{B})) \ge 0, \forall X, \rho_{\kappa^{*}}(\kappa^{*}(X_{A}),\kappa^{*}(X_{B})) = 0 \implies X_{A} = 0.
\end{aligned}
\end{gather}
Then for any $\lambda \in [0,1]$ holds the following conjunction of sub-additivity and the axiom of indiscernibles:
\begin{equation}
\begin{aligned}
\label{eq:convexity}
\rho_{\kappa^{*}}(\lambda \kappa(X_{A}) + (1 - \lambda)\kappa(X_{B})) & \le \rho_{\kappa} (\lambda \kappa(X_{A}),\kappa(I_{n})) + \rho_{\kappa} ((1 - \lambda)\kappa(X_{B}),X_{B})\\
                                                                  & = \lambda \rho_{\kappa} (\kappa(X_{A}),\kappa(X_{B})) + (1 -\lambda)\rho_{\kappa} (\kappa(X_{A}),\kappa(X_{B})), \forall \lambda \in [0,1].
\end{aligned}
\end{equation}
\end{proof}

\begin{lemma}~\label{lem:cauchy}
% Next, we prove the existence of a unique finite upper-bound upon the Kemeny metric space for any collection of $n$ elements, whose lower bound is known to be 0 by the axiomatic properties of a metric space (Definition~\ref{def:metric_space}) and by Lemma~\ref{lem:lower_bound}. 
There exists a unique finite upper-bound upon the Kemeny metric space for any collection of \(n\) element length vectors in the extended reals, whose lower bound is 0.
\end{lemma}
\begin{proof}
For any $\kappa^{*}(X_{j})\; \forall X_{j}, j = 1,\ldots,p$ results a mapping vector of maximum $\frac{(n^{2}-n)}{2}$, to which are assigned one of three distinct values: $a_{ij} \in \{1,-1,0\}$. If all $X_{i},\, i = 1,\ldots,n$ are in a monotonically ascending sequence, then there is \(M = n^{n}, \forall \infty^{-} < n < \infty^{+}\), representing all vectors of extended real scores, $n^{n} = \{n_{1}, n_{2}, \cdots, n_{n}\}$. Each vector, under operation $\kappa^{*}$ results a mapping of dimension $\frac{(n^{2}-n)}{2}$; the reduced set of mappings is granted by the symmetry of the metric space. It therefore follows that for fixed $n$ upon monotonic sequence \(X_{j}\) exists $ I_{n} = \{1,2,3,\ldots\} \in X_{n}$, form which results $\sum_{i=1}^{n}\kappa^{*}(X_{j}) = \sum \{1a,1a,\ldots,1a\}$. By equation~\ref{eq:kemeny_dist1} it then follows that upon such a field, the maximum distance attainable is found by the sequence of the absolute value of the difference for a similarly monotonically ordered sequence, for which is substituted the value  $1a \to -1a$, with exception only upon \(a_{i=j} = 0\): \[
\sum_{i=1}^{n}\sum_{j=1}^{n} |a^{*}_{ij} - (-1\cdot a^{*}_{ij})| = |\sum_{i=1}^{n} \{2_{i}a,\ldots,2_{n}a\}| = \frac{n^{2}-n}{2} \cdot 2a = a(n^{2}-n).\] Uniqueness proceeds as with the preceding Lemma, for which the summation over all mapping \(\kappa^{*}\) with one or more $a$ are by definition of equation~\ref{eq:kemeny_dist1} of distance greater than the observed, and therefore is uniquely extremised for any fixed finite positive $a$ and finite positive integer \(n\).%, and is therefore convergent within the domain of all extended real vectors within \([0,a(n^{2}-n)]\).
\end{proof}

\begin{lemma}
\label{lem:complete_space}
% Definition~\ref{def:complete_space} holds that a sequence in a metric space is Cauchy if for every positive real number $r > 0$ there is a positive integer $N$ such that for all positive integer vectors $X_{A},X_{B} \ge N$, $\rho(X_{A},X_{B}) < r.$ 
The Kemeny metric space is a complete metric space.
\end{lemma}
\begin{proof}
By the validity of both Lemma~\ref{lem:lower_bound} and Lemma~\ref{lem:cauchy} for the Kemeny metric, the maximum possible distance for said distance function, and any finite sequence $X$ of length $n$ is $a(n^{2}-n)$, as resulting from the additive sequence of $\frac{n^{2}-n}{2}$ elements, each valued $2a$, thereby establishing the definition off a positive and finite $r$ for all finite \(n\). Cauchy convergence therefore follows from the constructive mapping that for every $X_{A},X_{B} \in M \setminus \{I_{n},I_{N}^{\prime}\}$, there exists a finite distance $0 \le |\rho_{\kappa}(\kappa^{*}(X_{A}) - \kappa^{*}(X_{B}))| < a\cdot{r} < a(n^{2}-n)$ and it is shown that the Kemeny metric space is Cauchy convergent, per Definition~\ref{def:complete_space}.  An immediate corollary follows such that the upper and lower bounds of the Kemeny distance are determined by the conjunctive choice of finite $n$ and $a$, s.t., \[0\cdot{a} \le \rho_{\kappa}(a\kappa^{*}(X_{A}),\kappa^{*}(X_{B})) \le a(n^{2}-n),\] and the lemma is concluded, as the Kemeny metric space is shown to be complete for \(X_{n}\in \mathcal{X}\).
\end{proof}

% The separability and compactness (and therefore the existence of a Borel-$\sigma$ algebra) of the Kemeny metric space is now shown, wherein a separable space is a complete metric space which is also totally bounded:

\begin{theorem}~\label{lem:cauchy-schwarz}
% From Definition~\ref{def:cauchy-schwarz}, we prove that the Kemeny Banach norm-space satisfices the Cauchy-Schwarz inequality, and is therefore shown to be a continuous function.
The Kemeny complete metric space satisfies the Cauchy-Schwarz inequality and hence is a continuous function.
\end{theorem}
\begin{proof}
We first discuss the Cauchy-Schwarz inequality, and the existence of its lower-boundedness at 0. Let $\langle\mathbf{u},\mathbf{v}\rangle \in M;n$ be vectors resulting from the mapping function $\kappa$ whose length is $n$. By property (1) of Definition~\ref{def:metric_space} and Lemma~\ref{lem:lower_bound}, it follows that \[0 \le \langle c{\mathbf{u}}+\mathbf{v},c\mathbf{u}+\mathbf{v}\rangle = |c|^{2}\|\mathbf{u}\|^{2} + \|\mathbf{v}\|^{2} + c\langle \mathbf{v},\mathbf{u}\rangle + \bar{c}\langle \mathbf{u},\mathbf{v}\rangle,\]
for arbitrary $c \in \mathbb{C}$. If $\mathbf{u} \ne \vec{b}, \forall\ b \in \overline{\mathbb{R}}$, the affine linear transformation $c = -\langle{\mathbf{u}},\mathbf{v}\rangle/\|\mathbf{u}\|^{2}$ holds, from which immediately follows \[0 \le \|\mathbf{v}\|^{2} - \frac{\|\langle \mathbf{u},\mathbf{v}\rangle\|^{2}}{\|\mathbf{u}\|^{2}},\] which implies the Cauchy-Schwarz inequality; as equality is only found when $c\mathbf{u} + \mathbf{v} =0$, thereby defining a linear dependency. Under conditions $\rho_{\kappa}(x,y) = y = cx$ or $cy$, the Cauchy-Schwarz inequality is easily seen to hold. We next proceed to prove the norm $\|\mathbf{u}\|$ as defined in Definition~\ref{def:cauchy-schwarz}. Notice that by the identity of indiscernibles of a metric function space follows $\rho_{\kappa}(\kappa(X_{A}),\kappa(X_{B})) = 0 \therefore X_{A} = X_{B}$, and therefore by $\kappa^{\mathbf{u}}$ it follows that the minimum of 0 is obtained solely and uniquely upon the $n$ points for which all elements in the vector are equal and tied. Moreover, let it be seen that $\|c\mathbf{u}\| = \sqrt{\langle c\mathbf{u},c\mathbf{u}\rangle} = \sqrt{c^{2}\langle \mathbf{u},\mathbf{u}\rangle} = c\|\mathbf{u}\| = c(\kappa(\mathbf{u})\kappa(\mathbf{u})) = c\kappa^{2}(\mathbf{u})$, and then from which follows sub-additivity:
\begin{align*}
\rho_{\kappa}^{2}(\mathbf{u},\mathbf{v}) & = \langle \mathbf{u} + \mathbf{v}, \mathbf{u} + \mathbf{v}\rangle  = \kappa^{2}(\mathbf{u}) + \kappa^{2}(\mathbf{v}) + 2 \rho_{\kappa}(\mathbf{u},\mathbf{v})\\
                                         &  \qquad \le \kappa^{2}(\mathbf{u}) + \kappa^{2}(\mathbf{v}) + 2\kappa(\mathbf{u})\kappa(\mathbf{v}) = (\kappa(\mathbf{u}) \cdot \kappa(\mathbf{v}))^{2}.
\end{align*}
\end{proof}

\begin{lemma}
% \label{lem:compact}
\label{lem:kem_bounded}
% By Definition~\ref{def:compact} 
The Kemeny metric space is compact and totally bounded.% as it is complete and totally bounded. 
\end{lemma}
\begin{proof}
The Kemeny metric is complete by Lemma~\ref{lem:kem_expansion}, and totally bounded, in that it possesses no point of finite $n$ reals which is outside the bounds of $0 \le \rho_{\kappa}(a\kappa^{*}(X_{A}),a\kappa^{*}(X_{B})) \le a(n^{2}-n), \forall\ 0 < a < \infty^{+}$ by Lemma~\ref{lem:complete_space}. Therefore the Kemeny metric space is shown to be both compact and complete and consequently, separably dense as well.
\end{proof}
% In Theorem~\ref{lem:hilbert}, we shall show the Kemeny metric space is pre-Hilbert, and therefore there exists with a scalar product, conditional to three conditions. 

\begin{theorem}
\label{lem:hilbert}
% From Definition~\ref{def:hilbert}, a Banach norm space such as 
The Kemeny norm space is a pre-Hilbert space with the property that for finite positive \(a\), holds the following equivalence:
% shown to be a pre-Hilbert metric space if for $a > 0$ follows the equivalence denoted by positive homogeneity such that 
\[\rho_{\kappa}(a\cdot\kappa^{*}(X_{A}),\kappa^{*}(X_{B})) = a\rho_{\kappa}(\kappa^{*}(X_{A}),\kappa^{*}(X_{B})).\] 
\end{theorem}
\begin{proof}
By scalar multiplication, holds the lower bound of $0\cdot{a} =0$, from Lemma~\ref{lem:lower_bound}, and the upper bound (Lemma~\ref{lem:cauchy}) exists with a real number for all finite $n$. The canonical form of the Kemeny distance function is therefore scaled by $a^{-1}a = 1$, and therefore the linear distances are always equivalent to the closed interval $[0,1,\cdots,(n^{2}-n)], \forall\, n\ge 0$ if \(a=1\), and which is otherwise proportional to this sequence by positive finite scalar \(a \in \{(0,\infty^{+})\}\setminus 1\). Thus it is shown that the Kemeny Banach norm-space is also positive homogeneous while lacking an inner-product construction, and is therefore a pre-Hilbert space. 
\end{proof}
\begin{corollary}
From Theorem~\ref{lem:hilbert} we have justified the valid existence of an inner-product formulation for the Kemeny pre-Hilbert metric space, by definition, which is provided in equation~\ref{eq:kem_dist}. 
% This formulation, in equation~\ref{eq:kemeny_dist} \begin{subequations}

% \begin{equation}
% \label{eq:kemeny_dist}
% \rho_{\kappa}(X_{A},X_{B}) = \frac{1}{n^{2}-n} \sum_{i=1}^{n}\sum_{j=1}^{n} a_{ij}(X_{A})b_{ij}(X_{B}), i,j = 1,\ldots,n
% \end{equation}

% \begin{equation}
% \label{eq:kemeny_score}
% a_{ij} = {
% \begin{dcases}
% \:a & \text{ if } x_{i} \ge x_{j}\\
% -a & \text{ if } x_{i} < x_{j}\\
% \:0 & \text{ if } x_{i} = x_{j}.
% \end{dcases}
% }.
% }.
% \end{equation}
% \end{subequations} 
% This is equivalent to equation~\ref{eq:kem_dist} and presumes, but does not prove, that the Kemeny metric is a pre-Hilbert space, by continuing the notational usage for which \(a_{ij}\) denotes any element in the matrix \(\kappa(X_{j})\). With an inner-product formula, the Kemeny metric space is therefore a Hilbert space as given in equation~\ref{eq:kem_dist}. 
\end{corollary}

\end{document}